\documentclass[%
 aip,
 jmp,%
 amsmath,amssymb,
]{revtex4-1}

\DeclareMathOperator{\Det}{det}
\DeclareMathOperator{\spn}{span}
\newcommand*{\D}{\mathcal{D}}
\newcommand*{\e}{\mathrm{e}}
\usepackage{microtype}
\usepackage{graphicx}
\usepackage[caption=false]{subfig}

\usepackage{physics}

\usepackage{algorithm}
\usepackage{algorithmicx}
\usepackage{algpseudocode}
\algrenewcommand\algorithmicrequire{\textbf{Input:}}
\algrenewcommand\algorithmicensure{\textbf{Output:}}
\usepackage{dsfont}
\expandafter\let\csname equation*\endcsname\relax
\expandafter\let\csname endequation*\endcsname\relax
\usepackage{amsmath}
\usepackage{amsfonts}
\usepackage{mathtools}

\usepackage{amsthm}

\newtheorem{thm}{Theorem}
\newtheorem{lem}[thm]{Lemma}

\newtheorem{defn}[thm]{Definition}

\newcommand*{\defeq}{\mathrel{\vcenter{\baselineskip0.5ex \lineskiplimit0pt
      \hbox{\scriptsize.}\hbox{\scriptsize.}}}%
      =}
      
\usepackage{tensor}

\usepackage{graphicx}
\usepackage{dcolumn}
\usepackage{bm}

\begin{document}


\title[Algorithms for SU$(n)$ boson realizations and $\D$-functions]{Algorithms for SU$(n)$ boson realizations and $\D$-functions}

\author{Ish Dhand}
\email{ishdhand@gmail.com}
 \homepage{http://ishdhand.me/}
 \affiliation{Institute for Quantum Science and Technology, University of Calgary, Alberta T2N~1N4, Canada}

\author{Barry C.~Sanders}%
 \affiliation{Institute for Quantum Science and Technology, University of Calgary, Alberta T2N~1N4, Canada}
 \affiliation{Hefei National Laboratory for Physical Sciences at Microscale and Department of Modern Physics, University of Science and Technology of China, Hefei, Anhui 230026, China}
 \affiliation{CAS Center for Excellence and Synergetic Innovation Center in Quantum Information and Quantum Physics, University of Science and Technology of China, Shanghai 201315, China}
 \affiliation{Program in Quantum Information Science, Canadian Institute for Advanced Research, Toronto, Ontario M5G~1Z8, Canada}

\author{Hubert de Guise}
 \affiliation{Department of Physics, Lakehead University, Thunder Bay, Ontario P7B~5E1, Canada}

\date{\today}

\begin{abstract}
Boson realizations map operators and states of groups to transformations and states of bosonic systems. 
We devise a graph-theoretic algorithm to construct the boson realizations of the canonical SU$(n)$ basis states, which reduce the canonical subgroup chain, for arbitrary $n$. 
The boson realizations are employed to construct $\D$-functions, which are the matrix elements of arbitrary irreducible representations, of SU$(n)$ in the canonical basis.
We demonstrate that our $\D$-function algorithm offers significant advantage over the two competing procedures, namely factorization and exponentiation. 
\end{abstract}

\keywords{Special unitary group, irreducible representation, boson realization}
\maketitle

\section{Introduction}

$\mathcal{D}$-functions of a group element are the entries of irreducible matrix representations (irreps) of the element.
$\mathcal{D}$-functions of the special unitary group SU$(2)$ are important in nuclear, atomic, molecular and optical physics~\cite{Varvsalovivc1989,Rose1995,Edmonds1996,Racah1943,Jacob1959,Alder1956,Moshinsky1962}.
SU$(1,1)$ is the prototypical non-compact semi-simple Lie group, and its $\mathcal{D}$-functions appear in connection with Bogolyubov transformations, squeezing and parametric downconversion~\cite{Ui1970,Yurke1986}.
Methods for construction of intelligent states and the analysis of cylindrical Laguerre-Gauss beams employ $\mathcal{D}$-functions of SU$(1,1)$~\cite{Joanis2010,Karimi2014}.
$\D$-functions of other Lie groups enable exact solutions to problems in quantum optics~\cite{Cervero1996,Wunsche2000}.

One recent application of $\D$-functions of SU$(n)$ for arbitrary $n$ is to the BosonSampling problem, which deals with SU$(n)$ transformations acting on indistinguishable single-photon pulse inputs~\cite{Aaronson2013a,Aaronson2013}.
Within the framework of BosonSampling and of multi-photon interferometry in general, $\D$-functions provide a deeper understanding of the permutation symmetries between the interfering photons.
For instance, SU$(3)$ $\D$-functions enable a symmetry-based interpretation of the action of a three-channel linear interferometer on partially-distinguishable single-photon inputs~\cite{Tan2013,Tillmann2014}.
Exploiting the permutation symmetries present in multi-photon systems reduces the cost of computing interferometer outputs in comparison to brute-force techniques~\cite{Guise2014}.

Two existing procedures for computing SU$(n)$ $\D$-function are based on factorization and on exponentiation.
Both procedures have drawbacks, which we describe as follows.
Factorization-based methods, which compute SU$(n)$ $\mathcal{D}$-functions in terms of $\D$-functions of subtransformations, are well developed for groups of low rank~\cite{Vilenkine1968,Miller1968,Talman1968,Chacon1966,Rowe1999}.
However, generalizing these algorithms to higher $n$ requires SU$(n-1)$ coupling and recoupling coefficients, which have limited availability for $n>3$, i.e., restricted to certain subgroups of SU$(3)$~\cite{Draayer1973,Millener1978,Rowe2000}.
Hence, methods for $\D$-functions of higher groups are underdeveloped despite the application of their corresponding algebras to diverse problems~\cite{Beg1965,Slansky1981,Georgi1999,Rowe2010}. 

Another approach for computing SU$(n)$ $\D$-functions involves exponentiating and composing the matrix representations of the algebra~\cite{Cornwell1997,Gilmore2012}.
This approach has three hurdles.
For one, this method requires knowledge of all the matrix elements of each generator to be exponentiated.
Certain applications require closed-form expressions of $\D$-functions in terms of elements of the fundamental representation; exponentiation-based methods are infeasible for these applications because of the difficulty of exponentiating matrices analytically, especially for $n>5$. 
Furthermore, if only a limited number of $\D$-functions are required, exponentiation is wasteful because it computes the entire set of $\D$-functions.

We overcome the shortcomings of these algorithms by utilizing boson realizations, which map the algebra and its carrier space to bosonic operators and spaces respectively.
Boson realizations arise naturally when considering the groups Sp($2n,\mathds{R}$), SU($n$) and some of their subgroups.
For instance, SU$(1,1)$, SU$(2)$ and SU$(3)$ boson realizations are used to study degeneracies, symmetries and dynamics in quantum systems~\cite{Jauch1940,Schwinger1952,Baker1956,Elliott1958,Fradkin1965,Iachello1987,Klein1991,Kuriyama2000,Bartlett2001}.
A wide class of problems in theoretical physics rely on boson realizations of the symplectic group~\cite{Hwa1966,Kramer1966,Moshinsky1971,Quesne1971,Rowe1984}.

Here we aim to devise an algorithm to construct the $\D$-functions of arbitrary representations of SU$(n)$ for arbitrary $n$. 
We approach the problem of limited availability of SU$(n)$ $\D$-functions~\cite{Weyl1950,Chaturvedi2006} by presenting
(i) a mapping of the weights of an irrep to a graph, (ii) a graph-theoretic algorithm to compute boson realizations of the canonical basis states of SU$(n)$ for arbitrary $n$ (Algorithm~\ref{Algorithm:Main} in Subsection~\ref{Subsec:BasisStates}) and (iii) an algorithm that employs the constructed boson realizations to compute expressions for $\D$-functions as polynomials in the matrix elements of the defining representation (Algorithm~\ref{Algorithm:D} in Subsection~\ref{Subsec:D}).

The rest of the paper is structured as follows.
Section~\ref{Section:Definitions} includes definitions of the SU$(n)$ operators and basis states.
In Section~\ref{Section:BosonRealizations}, we define SU$(n)$ boson realizations and illustrate the calculations of SU$(2)$ $\D$-functions using SU$(2)$ boson realizations. 
Section~\ref{Section:Algorithms} details our algorithms for boson realizations of SU$(n)$ basis states and for the SU$(n)$ $\D$-functions.
We discuss potential generalizations of our algorithms in Section~\ref{Section:Conclusion}.

\section{Background: The special unitary group and its algebra}
\label{Section:Definitions}
In this section, we recall the relevant properties of special-unitary group SU$(n)$ and its algebra $\mathfrak{su}(n)$.
We explain how the $\mathfrak{su}(n) \supset \mathfrak{su}(n-1) \supset \dots \supset \mathfrak{su}(2)$ subalgebra chain is used to label the basis states of the unitary irreps of SU$(n)$.
We present the background for $n =2$ in Subsection~\ref{Subsection:SU2} before dealing with SU$(n)$ for arbitrary $n$ in Subsection~\ref{Subection:SUn}.

\subsection{SU$(2)$ operators and basis states}
\label{Subsection:SU2}
Consider the special unitary group 
\begin{equation}
\mathrm{SU}(2) = \{V: V\in GL(2,\mathds{C}),\,V^\dagger V = \mathds{1},\, \Det{V} = 1\}
\end{equation}
of $2\times 2$ special unitary matrices. 
Each element of SU$(2)$ can be parametrized by three angles $\Omega=(\alpha,\beta,\gamma)$.
The defining $2\times 2$ representation of an element $V(\Omega)$ of SU$(2)$ is given by
\begin{align}
\mathrm{V}(\Omega)
&=
\begin{pmatrix}
	\text{e}^{-\frac{1}{2}\text{i}(\alpha +\gamma )} \cos \frac{\beta}{2} & -\text{e}^{-\frac{1}{2}\text{i}(\alpha -\gamma )} \sin \frac{\beta}{2} \\
	\text{e}^{\frac{1}{2}\text{i}(\alpha -\gamma )} \sin \frac{\beta}{2} & \text{e}^{\frac{1}{2}\text{i}(\alpha +\gamma )} \cos \frac{\beta}{2}
\end{pmatrix}.
\label{Eq:2x2Rmatrix}
\end{align}
The Lie algebra corresponding to group SU$(2)$ is denoted by $\mathfrak{su}(2)$ and is spanned by the operators $J_x,J_y,J_z$, which satisfy the angular momentum commutation relations 
\begin{equation}
[J_x,J_y] = iJ_z\, ,\quad
[J_y,J_z] = iJ_x\, ,\quad
[J_z,J_x] = iJ_y.
\label{Eq:Jxyz}
\end{equation}

We transform the basis~(\ref{Eq:Jxyz}) of $\mathfrak{su}(2)$ to the complex combinations
\begin{equation}
C_{1,2} = {J_x + i J_y}\, ,\qquad
C_{2,1} = {J_x - i J_y}\, ,\qquad
H_1 = 2J_z,
\label{Eq:PlusMinusZ}
\end{equation}
which satisfy the commutation relations
\begin{equation}
[H_1,C_{1,2}] = 2C_{1,2} \, ,\quad
[H_1,C_{2,1}] = -2C_{2,1}\, ,\quad
[C_{1,2},C_{2,1}] = H_1.
\label{Eq:SU2RaisingLoweringCommutations}
\end{equation}
These commutation relations~(\ref{Eq:SU2RaisingLoweringCommutations}) facilitate the construction of a $(2J+1)$-dimensional irrep with carrier space spanned by basis states $\{\ket{J,M}:\, -J\le M\le J\}$~\cite{Littlewood1950}. 
The integer $2M$ is the weight of the eigenstate $\ket{J,M}$ for
\begin{equation}
H_1\ket{J,M} = 2M\ket{J,M}.
\end{equation}
The operators $C_{1,2}$ and $C_{2,1}$ act on eigenstates of $H_1$ by raising or lowering the weight $2M$ of the states
\begin{align}
C_{1,2} \ket{J,M} = \sqrt{J(J+1)-M(M+1)}\ket{J,M+1},\\
C_{2,1} \ket{J,M} = \sqrt{J(J+1)-M(M-1)}\ket{J,M-1},
\end{align}
where $2J$ is the highest eigenvalue of $H_1$.

Each basis state of a finite-dimensional irrep of SU$(2)$ is labelled by integral weight $2M\in\{-2J,-2J+2,\dots,2J-2,2J\}$. 
The unique basis state $\ket{J,J}$ is called the highest-weight state (hws) and is annihilated by the action of the raising operator $C_{1,2}$.
The representation is labelled by the largest eigenvalue $2J$ of $H_1$.
Basis states of an SU$(2)$ irrep are visualized as collections of points on a line with the location of each point related to the weight of the state. 
Figure~\ref{Figure:SU2} gives a geometrical representation of the action of $\mathfrak{su}(2)$ operators and illustrative examples of SU$(2)$ irreps.


\begin{figure}	
    \centering
    \subfloat[]{\includegraphics[width=2.015in]{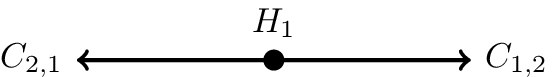}}%
    \qquad
    \subfloat[]{\includegraphics[width=3.15in]{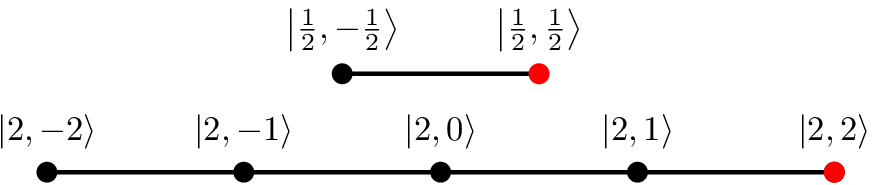}}%
\caption{(a) Generators of the $\mathfrak{su}(2)$ Algebra. 
The action of the raising and lowering operators $C_{1,2}, C_{2,1}$ on the basis states is represented by the directed lines. 
The basis states are invariant under the action of the Cartan operator $H_1$, which is represented by the dot at the centre. 
(b) SU$(2)$ irreps labelled by highest weights $2M=1$ and $2M=4$ respectively. 
The dots represent the basis states while the lines connecting the dots represent the transformation from one basis state to another by the action of the $\mathfrak{su}(2)$ raising and lowering operators. 
The red dot represents the hws, which is annihilated by the action of the raising operator $C_{1,2}$.}\label{Figure:SU2}
\end{figure}

\subsection{Basis states and $\D$-functions of SU$(n)$ for arbitrary $n$}
\label{Subection:SUn}
Next we consider the case of arbitrary $n$.
The unitary group U$(n)$ is the Lie group of $n\times n$ unitary matrices 
\begin{equation}
\mathrm{U}(n) \defeq \{V\colon \,V\in \mathrm{GL}(n,\mathds{C}), V^\dagger V = \mathds{1}\}.
\end{equation}
The corresponding Lie algebra is denoted by $\mathfrak{u}(n)$.
The complex extension of $\mathfrak{u}(n)$ is spanned by $n^2$ operators $\{C_{i,j} \colon i,j \in 1,2, \dots n\}$ 
satisfying the canonical commutation relations
\begin{equation}
 [C_{i,j},C_{k,l}] = \delta_{j,k}C_{i,l} - \delta_{i,l}C_{k,j}.
 \label{Eq:Commutation}
\end{equation}
The group SU$(n)$ is the subgroup of those U$(n)$ transformations that satisfy the additional property $\Det V=1$; i.e., 
\begin{equation}
\label{Eq:Determinant}
\mathrm{SU}(n) \defeq \{V\colon V \in \mathrm{U}(n),\Det V = 1\}.
\end{equation}
The U$(n)$ $\mathcal{D}$-functions differ from the SU$(n)$ 
$\mathcal{D}$-functions by at most
a phase, and we concentrate here on the SU$(n)$ case.

The operator $N= C_{1,1}+C_{2,2}+\dots + C_{n,n}$ is in the centre%
\footnote{The centre of an algebra $\mathfrak{u}$ comprises those elements $x$ of $\mathfrak{u}$ such that $xu = ux$ for all $u\in\mathfrak{u}$.} 
of $\mathfrak{u}(n)$.
The Lie algebra $\mathfrak{su}(n)$ is obtained from $\mathfrak{u}(n)$ by eliminating the operator $N$. 
The $n-1$ operators 
\begin{equation}
H_i = C_{i,i} - C_{i+1,i+1} \quad \forall i \in \{1,2,\dots,n-1\}
\label{Eq:Cartan}
\end{equation}
commute with each other and span the Cartan subalgebra of $\mathfrak{su}(n)$.
Hence, we have the following definition of the $\mathfrak{su}(n)$ algebra.
\begin{defn}[$\mathfrak{su}(n)$ algebra~\cite{Littlewood1950}]
The algebra $\mathfrak{su}(n)$ is the span of the operators $\{C_{i,j} \colon i,j \in \{1,2, \dots, n\},\, i\ne j\}$ and $\{H_i: H_i = C_{i,i} - C_{i+1,i+1},\, i \in \{1,2,\dots,n-1\}\}$ where the operators $\{C_{i,j}\}$ obey the commutation relations
\begin{equation}
[C_{i,j},C_{k,l}] = \delta_{j,k}C_{i,l} - \delta_{i,l}C_{k,j}.
\end{equation}
\label{Defn:Algebra}
\end{defn}
\noindent The linearly independent (LI) $\mathfrak{su}(n)$ basis states span the carrier space of $\mathfrak{su}(n)$ representations.
Each basis state is associated with a weight, which is the set of integral eigenvalues of the Cartan operators.
\begin{defn}[Weight of $\mathfrak{su}(n)$ basis states~\cite{Littlewood1950}]
The weight of a basis state is the set $\Lambda = (\lambda_1,\lambda_2,\dots,\lambda_{n-1})$ of $n-1$ integral eigenvalues of the Cartan operators $\{H_1,H_2,\dots,H_{n-1}\}$.
$\mathfrak{su}(n)$ basis states have well defined weights.
\end{defn}

Of the $n^2-1$ elements, $n-1$ Cartan operators generate the maximal Abelian subalgebra of $\mathfrak{su}(n)$.
The remaining operators satisfy the commutation relation
\begin{equation}
[H_{i},C_{j,k}] =
\begin{cases}
\beta_{i,jk} C_{j,k}, &\forall j < k,\\
-\beta_{i,jk} C_{j,k}, &\forall j > k,
\end{cases}
\end{equation}
for Cartan operators $H_i$ of Definition~\ref{Defn:Algebra} and for positive integral roots $\beta_{i,jk}$.
The operators $\{C_{j,k}\colon j<k\}$ define a set of raising operators.
The remaining off-diagonal operators $\{C_{j,k}\colon j>k\}$ are the $\mathfrak{su}(n)$ lowering operators.
Each irrep contains a unique state that has nonnegative integral weights $K = (\kappa_1,\dots,\kappa_{n-1})$ and is annihilated by all raising operators.
This state is the hws of the irrep.
\begin{defn}[Highest-weight state] 
The hws of an SU$(n)$ irrep is the unique state that is annihilated according to 
\begin{equation}
C_{i,j}\ket{\psi^{K}_\text{hws}} = 0 \quad \forall i<j,\, i,j \in \{1,2,\dots,n\} 
\end{equation}
by the action of all the raising operators.
\end{defn}
\noindent 
The weight of the hws also labels the irrep; i.e., two irreps with the same highest weight are equivalent and two equivalent representations have the same highest weight.
Hence, we label an irrep by $K = (\kappa_1,\kappa_2,\dots,\kappa_{n-1})$ if the hws of the irrep has weight $\Lambda = K$.

Whereas in SU$(2)$ the weight $2M$ and the representation label $J$ are enough to uniquely identify a state in the representation, this is not so for SU$(n)$ representations. 
In general, more than one SU$(n)$ basis state of an irrep could share the same weight. 
For example, certain states of the $K= (2,2)$ irrep of SU$(3)$ irrep have the same weight~(Fig.~\ref{Figure:SU3}). 
The number of basis states that share the same SU$(n)$ weight $\Lambda = (\lambda_1,\lambda_2,\dots,\lambda_{n-1})$ is the multiplicity $M({\Lambda})$ of the weight~\cite{Kostant1959}.
Hence, uniquely labelling the SU$(n)$ basis states requires a scheme to lift the possible degeneracy of weights.

\begin{figure}	
    \centering
    \subfloat[]{\includegraphics[width=2.015in]{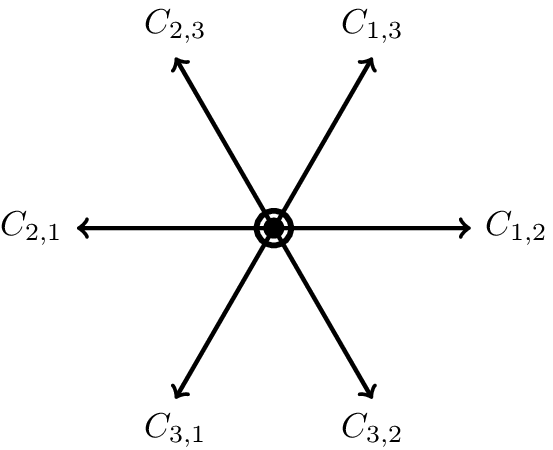}}%
    \qquad
    \subfloat[]{\includegraphics[width=3.51in]{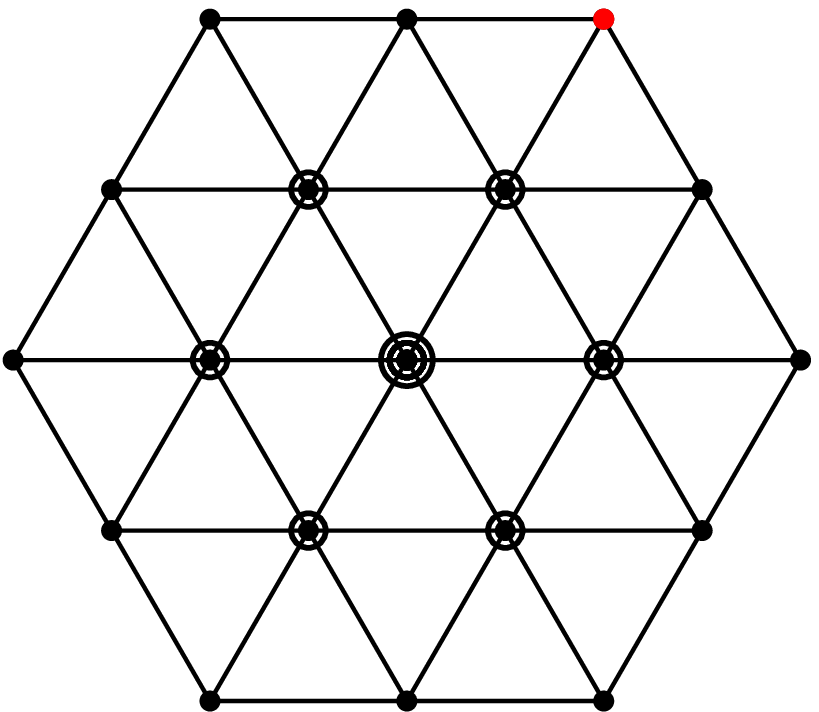}}%
\caption{(a) Generators of the $\mathfrak{su}(3)$ algebra. 
The action of the raising operators $\{C_{1,2}, C_{1,3}, C_{2,3}\}$ and lowering operators $\{C_{2,1},C_{1,3},C_{2,3}\}$ on the canonical basis states and their linear combinations is represented by the directed lines. 
(b) The SU$(3)$ irrep labelled by its highest weight $(\kappa_1,\kappa_2) = (2,2)$. 
The dots and circles represent the canonical basis states.
The dimension of the space of states at a given vertex is the sum of the number of dots and the number of circles at the vertex, for instance weights associated with dimension two are represented by one dot and one circle.
The lines connecting the dots represent the transformation from states of one weight to those of another by the action of SU$(3)$ raising and lowering operators. 
The red dot represents the highest weight of the irrep.
A unique hws occupying this weight is annihilated by the action of each of the raising operator.}
\label{Figure:SU3}
\end{figure}

One approach to labelling the SU$(n)$ basis states involves specifying the transformation properties under the action of the subalgebras of $\mathfrak{su}(n)$. 
We restrict our attention to the canonical subalgebra chain 
\begin{equation}
\mathfrak{su}_{1,2,\dots, n}(n)\supset \mathfrak{su}_{1,2,\dots, n-1}(n-1) \supset \dots \supset \mathfrak{su}_{1,2}(2),
\label{Eq:SubalgebraChain}
\end{equation}
where $\mathfrak{su}_{1,2,\dots, m}(m)$ is the subalgebra generated by the operators $\{C_{i,j}\colon i,j\in\{1,2,\dots,m\}\,,i\ne j\}$ and $\{H_k\colon k\in\{1,2,\dots,m-1\}\}$. 
Details about the choice of subalgebra chain are presented in~\ref{Appendix:SubAlgebraChoice}.
Henceforth, we drop the subscript and denote $\mathfrak{su}_{1,2,\dots, m}(m)$ by $\mathfrak{su}(m)$.

The canonical basis comprises the eigenstates of the $\mathfrak{su}(m)$ generators for all $m\le n$ according to the following definition.
\begin{defn}(Canonical basis states)
\label{Definition:CanonicalBasisStates}
The canonical basis states of SU$(n)$ irrep $K^{(n)}$ are those states
\begin{equation}
\Big|\tensor*{\psi}{*^{K^{(n)}}_{\Lambda^{(n)}}^{,\dots,}_{,\dots,}^{K^{(3)},}_{\Lambda^{(3)},}^{K^{(2)}}_{\Lambda^{(2)}}}\Big\rangle
\end{equation}
that have well defined values of 
\begin{enumerate}
\item
irrep labels $K^{(m)}$ for $\mathfrak{su}(m)$ algebras for all $\{m:2\le m\le n\}$ and
\item
$\mathfrak{su}(m)$ weights $\Lambda^{(m)}$, i.e., eigenvalues of the Cartan operators of $\mathfrak{su}(m)$ algebras for all $\{m:2\le m\le n\}$.
\end{enumerate}
\end{defn}
\noindent
Consider the example of the $(\kappa_1,\kappa_2) = (1,1)$ irrep of SU$(3)$.
There are two basis states with the weight $(\lambda_1,\lambda_2) = (0,0)$.
We can identify these two states by specifying
\begin{enumerate}
\item
the $\mathfrak{su}(3)$ irrep label $K^{(3)} = (\kappa_1,\kappa_2) = (1,1)$ and the $\mathfrak{su}(2)$ irrep label $K^{(2)} = (\kappa_1) = (0)$ or $K^{(2)} = (\kappa_1) = (1)$.
\item
the $\mathfrak{su}(3)$ weights $\Lambda^{(3)} = (\lambda_1,\lambda_2) = (0,0)$ and $ \mathfrak{su}(2)$ weight $\Lambda^{(2)} = (\lambda_1) = (0)$.
\end{enumerate}
The connection between our labelling of canonical basis states of Definition~\ref{Definition:CanonicalBasisStates} and the Gelfand-Tsetlin patterns~\cite{Gelfand1988} is detailed in~\ref{Appendix:Connection}.
The canonical basis state $\Big|\tensor*{\psi}{*^{K^{(n)}}_{\Lambda^{(n)}}^{,\dots,}_{,\dots,}^{K^{(3)},}_{\Lambda^{(3)},}^{K^{(2)}}_{\Lambda^{(2)}}}\Big\rangle
$ for which $K^{(m)} = \Lambda^{(m)}$ for all $m\in\{2,\dots,n\}$ is the highest weight of the irrep $K^{(n)}$.

The relative phases between the canonical basis states are fixed by comparing with the phase of the hws~\cite{Gelfand1988}.
Matrix elements of the simple raising operators $C_{\ell,\ell+1},\, \ell \in\{1,\dots,n-1\}$ are set as positive~\cite{Barut1986}.
Thus, we impose the following additional constraint on the canonical basis states
\begin{equation}
\mel**{\psi_\mathrm{hws}}{c_{1,2}^{p_{1,2}} c_{2,3}^{p_{2,3}}\cdots c_{n-1,n}^{p_{n-1,n}}}{\tensor*{\psi}{*^{K^{(n)}}_{\Lambda^{(n)}}^{,\dots,}_{,\dots,}^{K^{(3)},}_{\Lambda^{(3)},}^{K^{(2)}}_{\Lambda^{(2)}}}} \ge 0, 
\label{Eq:PhaseConvention}
\end{equation}
for all canonical basis states, for positive integers $p_{\ell,\ell+1}$.


$\D$-functions are the matrix elements of SU$(n)$ irreps. 
The rows and columns of SU$(n)$ matrix representations are labelled by SU$(n)$ basis states.
The expression for SU$(n)$ $\D$-functions generalize those of the SU$(2)$ $\D$-functions~(\ref{Eq:D}) with $M,M^\prime$ replaced by suitable labels for weights and $J$ replaced by suitable subalgebra labels.
\begin{defn}[$\D$-functions]
$\D$-functions of an SU$(n)$ transformation $V(\Omega)$ are 
\begin{equation}
\tensor*{\D}{*^{K^{(n)}}_{\Lambda^{(n)}}^{,\dots,}_{,\dots,}^{K^{(3)},}_{\Lambda^{(3)},}^{K^{(2)}}_{\Lambda^{(2)}}^;_;^{K^{\prime(n)}}_{\Lambda^{\prime(n)}}^{,\dots,}_{,\dots,}^{K^{\prime(3)},}_{\Lambda^{\prime(3)},}^{K^{\prime(2)}}_{\Lambda^{\prime(2)}}}
(\Omega) 
\defeq \Big\langle{\tensor*{\psi}{*^{K^{(n)}}_{\Lambda^{(n)}}^{,\dots,}_{,\dots,}^{K^{(3)},}_{\Lambda^{(3)},}^{K^{(2)}}_{\Lambda^{(2)}}}}\Big|V(\Omega)\Big|{\tensor*{\psi}{*^{K^{\prime(n)}}_{\Lambda^{\prime(n)}}^{,\dots,}_{,\dots,}^{K^{\prime(3)},}_{\Lambda^{\prime(3)},}^{K^{\prime(2)}}_{\Lambda^{\prime(2)}}}\Big\rangle
},
\label{Eq:Defined}
\end{equation}
where $\Omega = \{\omega_1,\omega_2,\dots,\omega_{n^2-1}\}$ is the set of $n^2-1$ independent angles that parameterize an SU$(n)$ transformation~\cite{Reck1994}. 
\end{defn}
\noindent Note that SU$(n)$ $\D$-functions~(\ref{Eq:Defined}) are non-zero only if the left and the right states belong to the same SU$(n)$ irrep, i.e., 
\begin{equation}
K^{(n)} \ne K^{\prime(n)}\implies \,\tensor*{\D}{*^{K^{(n)}}_{\Lambda^{(n)}}^{,\dots,}_{,\dots,}^{K^{(3)},}_{\Lambda^{(3)},}^{K^{(2)}}_{\Lambda^{(2)}}^;_;^{K^{\prime(n)}}_{\Lambda^{\prime(n)}}^{,\dots,}_{,\dots,}^{K^{\prime(3)},}_{\Lambda^{\prime(3)},}^{K^{\prime(2)}}_{\Lambda^{\prime(2)}}}
(\Omega) = 0.
\end{equation} 
$\D$-functions of an irrep $K$ refer to those $\D$-functions for which $K^{(n)} = K^{\prime(n)} = K$.

We approach the task of constructing SU$(n)$ $\D$-functions by using boson realizations of SU$(n)$ states.
In the next section, we define boson realizations and illustrate the construction of SU$(2)$ $\D$-functions using SU$(2)$ boson realizations.

\section{Background: boson realizations of SU$(n)$}
\label{Section:BosonRealizations}
In this section, we describe boson realizations, which map $\mathfrak{su}(n)$ operators and carrier-space states to operators and states of a system of $n-1$ species of bosons on $n$ sites respectively.
We first present the mapping for $n=2$ and illustrate SU$(2)$ $\D$-functions calculation using the SU$(2)$ boson realization in Subsection~\ref{Subsec:SU2BR}.
Boson realizations of SU$(n)$ for arbitrary $n$ are defined is Subsection~\ref{Subsec:SUnBR}.

\subsection{SU$(2)$ boson realizations}
\label{Subsec:SU2BR}
The commutation relations~(\ref{Eq:SU2RaisingLoweringCommutations}) of $\{C_{1,2},C_{2,1},H_1\}$ are reproduced by number-preserving bilinear products of creation and annihilation operators that act on a two-site bosonic system.
Specifically, the $\mathfrak{su}(2)$ operators have the boson realization
\begin{equation}
 C_{1,2} \mapsto c_{1,2} \defeq a_1^\dagger a_2\, ,\quad
C_{2,1} \mapsto c_{2,1} \defeq a_2^\dagger a_1 \, ,\quad
H_1 \mapsto h_1 \defeq a_1^\dagger a_1 - a_2^\dagger a_2,
\label{Eq:su2bosonmap}
\end{equation}
where the bosonic creation and annihilation operators obey the commutation relations
\begin{equation}
\left[a_i,a_j^\dagger\right] = \delta_{ij}\mathds{1}, \quad \left[a_i,a_j\right] = \left[a_i^\dagger,a_j^\dagger\right] = 0.
\label{Eq:ccr}
\end{equation}
Here and henceforth, we use lower-case symbols for boson realizations of the respective upper-case symbols.
Explicitly, 
\begin{equation}
[h_1,c_{1,2}] = 2c_{1,2}\,\qquad
[h_1,c_{2,1}] = - 2c_{2,1}\,\qquad
[c_{1,2},c_{2,1}] = h_1.
\end{equation}
The operators $\{c_{1,2},c_{2,1},h_1\}$ also span the complex extension of the $\mathfrak{su}(2)$ Lie algebra.

Boson realizations map the states in the carrier space of SU$(2)$ to the states of a two-site bosonic system.
Specifically, each basis state of the $(2J+1)$-dimensional SU$(2)$ irrep maps 
\begin{equation}
\ket{J,M} \mapsto \frac{(a_1^\dagger)^{J+M} (a_2^\dagger)^{J-M} }{\sqrt{(J+M)!(J-M)!}}
\ket{0}
\label{Eq:su2polynomialstate}
\end{equation}
to the state of a two-site system with $J+M$ and $J-M$ bosons in the two sites respectively.

The $(2J+1)$-dimensional irreps of SU$(2)$ map to number-preserving transformations on a two-site system of $2J$ bosons in the basis of Eq.~(\ref{Eq:su2polynomialstate}).
The elements of these $(2J+1)\times(2J+1)$ matrices are the SU$(2)$ $\D$-functions
\begin{equation}
D^J_{M^\prime M}(\Omega)
		\defeq \bra{J,M^\prime}V(\Omega)\ket{J,M}
\label{Eq:D}
\end{equation}
for irrep $J$ and row and column indices $M^\prime, M$.
The expression for $\D$-functions~(\ref{Eq:D}) of SU$(2)$ element $V(\Omega)$ can be calculated by noting that the creation operators transform under the action of $V$ of Eq.~(\ref{Eq:2x2Rmatrix}) according to
\begin{align}
a_1^\dagger\to V_{11}a_1^\dagger + V_{12}a_2^\dagger, \nonumber \\
a_2^\dagger\to V_{21}a_1^\dagger + V_{22}a_2^\dagger,
\end{align}
where $V$ is the $2\times 2$ fundamental representation of $V(\Omega)$.
The state $\ket{J,M}$~(\ref{Eq:su2polynomialstate}) thus transforms to
\begin{equation}
\ket{J,M}\to \frac{\left(V_{11}a_1^\dagger + V_{12}a_2^\dagger\right)^{J+M}\left(V_{21}a_1^\dagger + V_{22}a_2^\dagger\right)^{J-M}}
{\sqrt{(J+M)!(J-M)!}}\ket{0}\,
\label{Eq:transformedsu2polynomialstate} 
\end{equation}
as the vacuum state $\ket{0}$ is invariant under the action $V$.
Using Eqs.~(\ref{Eq:su2polynomialstate}) and (\ref{Eq:transformedsu2polynomialstate}), we obtain
\begin{equation}
D^J_{M^\prime M}(\Omega) =\Bigg\langle0\Bigg|\frac{a_1^{J+M'}a_2^{J-M'}\left(V_{11}a_1^\dagger + V_{12}a_2^\dagger\right)^{J+M}\left(V_{21}a_1^\dagger + V_{22}a_2^\dagger\right)^{J-M}}{\sqrt{(J+M')!(J-M')!}\sqrt{(J+M)!(J-M)!}}\Bigg|0\Bigg\rangle,
\label{Eq:DFunction}
\end{equation}
which can be evaluated using the commutation relations of the creation and annihilation operators~(\ref{Eq:ccr}).%
\footnote{
A useful computational shortcut involves the map
\begin{equation}
a_k^\dagger\to x_k\, ,\quad a_\ell\to \frac{\partial}{\partial x_\ell}\,,\quad k,\ell \in\{1,2\},
\label{Eq:PolynomialMap}
\end{equation}
which preserves the boson commutation relations.
The map~(\ref{Eq:PolynomialMap}) transforms the vector~$\ket{J,M}$~(\ref{Eq:su2polynomialstate}) into a formal polynomial and the corresponding dual vector~$\bra{J,M}$ into a linear differential operator in the dummy variables $x_1,x_2$.
The $\D$-function~(\ref{Eq:DFunction}) is thus evaluated as the action of a linear differential operator on a polynomial in $x_1,x_2$.
}

In this paper, our objective is to generalize Eqs.~(\ref{Eq:su2bosonmap}) and (\ref{Eq:su2polynomialstate}) systematically from $n = 2$ to arbitrary $n$.
In the next subsection, we define boson realizations of operators and carrier-space states of $\mathfrak{su}(n)$.
Furthermore, we construct the boson realization for the hws of arbitrary SU$(n)$ irreps.

\subsection{SU$(n)$ boson realizations for arbitrary $n$}
\label{Subsec:SUnBR}
SU$(n)$ boson realizations map SU$(n)$ states $\Big|{\tensor*{\psi}{*^{K^{(n)}}_{\Lambda^{(n)}}^{,\dots,}_{,\dots,}^{K^{(3)},}_{\Lambda^{(3)},}^{K^{(2)}}_{\Lambda^{(2)}}}}\Big\rangle$ and $\mathfrak{su}(n)$ operators to states and operators of a system of bosons on $n$ sites. 
Bosons are labelled based on the site $i\in \{1,2,\dots,n\}$ at which they are situated and by an internal degree of freedom, which is denoted by an additional subscript on the bosonic operators. 
The bosonic creation and annihilation operators on this system are
\begin{align}
&~\Big\{a^\dagger_{i,j}\colon i\in \{1,2,\dots,n\}, j \in \{1,2,\dots,n-1\} \Big\} \quad \text{Creation}\\
&~\Big\{a_{k,l}\colon k\in \{1,2,\dots,n\}, l\in \{1,2,\dots,n-1\}\Big\}\quad \text{Annihilation},
\end{align}
where the first label in the subscript is the usual index of the site occupied by the boson.
The second index refers to the internal degrees of freedom of the boson.
Each boson can have at most $n-1$ possible internal states to ensure that basis states can be constructed for arbitrary irreps.
In photonic experiments, this internal degree of freedom could correspond to the polarization, frequency, orbital angular momementum or the time of arrival of photons.

The $\mathfrak{su}(n)$ operators are mapped to number-preserving bilinear products of boson creation and annihilation operators. 
Specifically, raising and lowering operators $C_{i,j}$ of $\mathfrak{su}(n)$ map to bosonic operators $c_{i,j}$ according to
\begin{equation}
C_{i,j} \mapsto c_{i,j} \defeq \sum_{k=1}^{n-1} a^\dagger_{i,k} a_{j,k}.
\end{equation} 
Operators $\{c_{i,j}\}$ make bosons hop from site $j$ to site $i$.
The operators $h_i$ are the image of the Cartan operators $H_i$:
\begin{equation}
H_{i} \mapsto h_{i} \defeq a^\dagger_i a_i - a^\dagger_{i+1} a_{i+1}.
\end{equation}
Operators $\{h_i\}$ count the difference in the total number of bosons at two sites and commute among themselves.
As usual, we used the upper-case symbols to denote the $\mathfrak{su}(n)$ elements and the corresponding lower-case symbols for the respective boson operators.

The boson realizations of the basis states of SU$(n)$ are obtained by the action of polynomials in creation operators $\{a^\dagger_{i,j}\colon i\in \{1,2,\dots,n\},j\in \{1,2,\dots,n-1\}\}$ on the $n$-site vacuum state $\ket{0}$. 
Each term in the polynomial is a product of 
\begin{equation}
N_K= \kappa_1 + 2\kappa_2 + \dots + (n-1)\kappa_{n-1}
\label{Eq:N}
\end{equation}
boson creation operators for basis states in irreps $K = (\kappa_1,\kappa_2,\dots,\kappa_{n-1})$.
Therefore, an SU$(n)$ basis state is specified by the coefficient of a polynomial consisting of terms that are products of $N_K$ creation operators.

The hws of a given SU$(n)$ irrep can be explicitly constructed in the boson realization (as polynomials in creation and annihilation operators) according to the following lemma.
\ignorespaces
\begin{lem}[Boson realization of hws~\cite{Moshinsky1962a,Moshinsky1962}]
\label{Lemma:hws}
The bosonic state 
\begin{equation}
\ket{\psi^K_{\rm hws}}
= 
\Det \begin{pmatrix}a_{1,1}^\dagger&\dots&a_{1,n-1}^\dagger\\
\vdots&\ddots&\vdots\\
a_{n-1,1}^\dagger&\dots &a_{n-1,n-1}^\dagger\end{pmatrix}^{\kappa_{n-1}}\cdots 
\Det \begin{pmatrix}a_{1,1}^\dagger&a_{1,2}^\dagger\\a_{2,1}^\dagger&a_{2,2}^\dagger\end{pmatrix}^{\kappa_2}
\Det \begin{pmatrix}{a_{1,1}^\dagger}\end{pmatrix}^{\kappa_1}\ket{0}
\label{Eq:Hws}
\end{equation}
is a hws for a given SU$(n)$ irrep $K = (\kappa_1,\kappa_2,\dots,\kappa_{n-1})$.
\end{lem}
\noindent \ignorespaces
One can verify that the state $\ket{\psi^K_\text{hws}}$~(\ref{Eq:Hws}) is annihilated
\begin{equation}
c_{j,k} \ket{\psi^K_\text{hws}} = 0 \quad \forall j<k
\end{equation}
by the action of any of the raising operators.

Thus, the hws of any irrep can be constructed analytically using Lemma~\ref{Lemma:hws}.
In the following section, we provide an algorithm to construct each of the basis states of arbitrary SU$(n)$ irreps.
Furthermore, we present an algorithm to compute expressions for SU$(n)$ $\D$-functions in terms of the entries of the fundamental representation.

\section{Results: Algorithms for boson realizations of SU$(n)$ states and for SU$(n)$ $\D$-functions}
\label{Section:Algorithms}

In this section, we present three algorithms%
\footnote{
In Algorithms~\ref{Algorithm:BasisSet}-\ref{Algorithm:D}, we denote operations in capital case and \textsc{SmallCaps} font. 
Variables are denoted by roman font and are in lower case.
}.
Algorithm~\ref{Algorithm:BasisSet} (boson-set algorithm) constructs basis sets for each weight of a given $\mathfrak{su}(n)$ irrep.
Algorithm~\ref{Algorithm:Main} (canonical-basis-states algorithm) employs the boson-set algorithm to compute expressions for the canonical basis states of a given SU$(n)$ irrep $K$.
The states thus constructed are used by Algorithm~\ref{Algorithm:D} to calculate the $\D$-functions of a given SU$(n)$ transformation.

\begin{figure}
\centering
\includegraphics[width=3in]{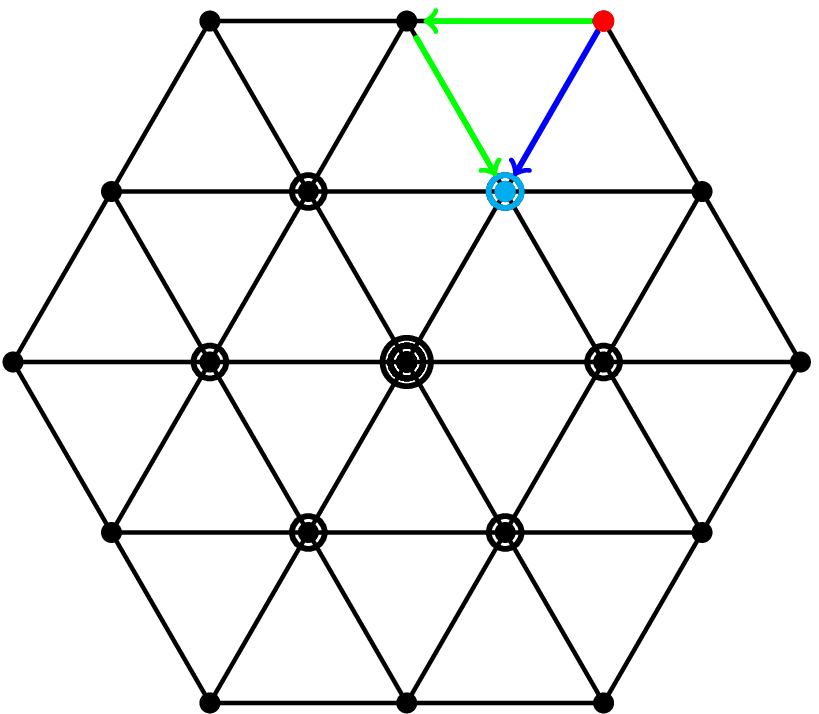}\\
\caption{
The first step of the basis-set computation algorithm~(Algorithm~\ref{Algorithm:BasisSet}) illustrated for the $(2,2)$ irrep of SU$(3)$, where the dimension of the space of states at a given vertex is the sum of the number of dots and the number of circles at the vertex.
The algorithm constructs the hws (occupying the red vertex) using Lemma~\ref{Lemma:hws}.
The lowering operators can transform states at one vertex to states at another vertex along different paths connecting the starting and the target vertex, for instance the two paths coloured green and blue.
Lowering along the different paths to reach a target vertex will generate the same number of LI as the weight multiplicity.
In our illustration, we obtain a basis set that contains two independent states at the target vertex.
The algorithm traverses the irrep graph systematically until all basis sets are calculated.}%
\label{Figure:Algorithm}%
\end{figure}

Algorithms~\ref{Algorithm:BasisSet} and~2 rely on mapping the SU$(n)$ irrep to a graph and systematically traversing the graph to obtain basis states.
The vertices of the irrep graph are identified with the weights of the given irrep of SU$(n)$ and the edges with the action of the elements of the Lie algebra $\mathfrak{su}(n)$ on the states.
Specifically, the irrep graph $G = (\mathcal{V},\mathcal{E})$ of an SU$(n)$ irrep is defined as follows. 

\begin{defn}[Irrep graph]
The bijection 
\begin{equation}
v\colon \{\Lambda_1,\Lambda_2,\dots,\Lambda_{d} \} \to \mathcal{V}
\end{equation}
maps the set $\{\Lambda_1,\Lambda_2,\dots,\Lambda_d\}$ of the $d$ weights in the given irrep to the vertices 
\begin{equation}
\mathcal{V} = \{v(\Lambda_1),v(\Lambda_2),\dots,v(\Lambda_d)\}
\end{equation}
of its irrep graph.
Vertices $v(\Lambda_{k})$ and $v(\Lambda_{\ell})$ are connected by an edge $e_j = (v(\Lambda_{k}),v(\Lambda_{\ell})) \in \mathcal{E}$ iff $\exists~c_{i,j}, \Lambda_{k}, \Lambda_{\ell}$ such that 
\begin{equation}
c_{i,j\ne i}\ket{\psi_{\Lambda_{k}}} = \ket{\psi_{\Lambda_{\ell}}},
\end{equation}
where $\ket{\psi_{\Lambda_{k}}}$ and $\ket{\psi_{\Lambda_{\ell}}}$ are SU($n$) states that have weights ${\Lambda_{k}}$ and ${\Lambda_{\ell}}$ respectively. 
In general, states $\ket{\psi_{\Lambda_{k}}}$ and $\ket{\psi_{\Lambda_{\ell}}}$ are linear combinations of canonical basis states.
Edges~$\mathcal{E}$ together with the vertices $\mathcal{V}$ define the irrep graph $G = (\mathcal{V},\mathcal{E})$. 
\end{defn}

More than one basis state can have the same weight.
The number of basis states sharing a weight $\Lambda_i$ is defined as the multiplicity $M(\Lambda_i)$ of the weight.
In other words, each vertex $v(\Lambda_i)$ is identified with an $M(\Lambda_i)$-dimensional space spanned by those canonical basis states that have weight $\Lambda_i$.
The vertex space and vertex basis sets are defined as follows.
\begin{defn}[Vertex spaces]
The vertex space of $v(\Lambda_i)$ is the span
\begin{equation}
\Psi(\Lambda_i) = \spn{\left(\ket{{\psi^{1}_{\Lambda_i}}},\ket{\psi^{2}_{\Lambda_i}},\dots,\ket{\psi^{M(\Lambda_i)}_{\Lambda_i}}\right)}
\label{Eq:Space}
\end{equation}
of the canonical basis states~(Definition~\ref{Definition:CanonicalBasisStates}) that have the weight $\Lambda_i$.
\end{defn}
\noindent 
The set $\left\{\ket{\psi^{(1)}_{\Lambda_i}},\ket{\psi^{(2)}_{\Lambda_i}},\dots, \ket{\psi^{(M(\Lambda_i))}_{\Lambda_i}}\right\}$ of canonical basis states is not the only set that spans the vertex space $\Psi(\Lambda_i)$ of $v(\Lambda_i)$.
In general, basis sets of $\Psi(\Lambda_i)$ can be defined as follows.
\begin{defn}[Vertex basis sets]
The set 
\begin{equation}
\left\{\ket{\phi^{(1)}_{\Lambda_i}},\ket{\phi^{(2)}_{\Lambda_i}},\dots, \ket{\phi^{(M(\Lambda_i))}_{\Lambda_i}}\right\}
\end{equation} is called the basis set of a vertex $v(\Lambda_i)$ if it spans the vertex space $\Psi(\Lambda_i)$~(\ref{Eq:Space}) of $v(\Lambda_i)$, i.e., 
\begin{equation}
\spn{\left(\ket{\phi^{1}_{\Lambda_i}},\ket{\phi^{2}_{\Lambda_i}},\dots,\ket{\phi^{M(\Lambda_i)}_{\Lambda_i}}\right)}
= \Psi(\Lambda_i).
\label{Eq:Space2}
\end{equation}
\end{defn}
\noindent 
The states $\left\{\ket{\phi^{(1)}_{\Lambda_i}},\ket{\phi^{(2)}_{\Lambda_i}},\dots, \ket{\phi^{(M(\Lambda_i))}_{\Lambda_i}}\right\}$ are linear combinations of the canonical basis states $\left\{\ket{\psi^{(1)}_{\Lambda_i}},\ket{\psi^{(2)}_{\Lambda_i}},\dots, \ket{\psi^{(M(\Lambda_i))}_{\Lambda_i}}\right\}$.
Algorithm~\ref{Algorithm:BasisSet} computes basis sets of the spaces $\Psi(\Lambda_i)$ for each of the $d$ weights $\Lambda_i$ that occurs in a given irrep. 

\begin{figure}
\linespread{1}
\begin{algorithm}[H]
\caption{Basis-Set Algorithm}
\label{Algorithm:BasisSet}
	\begin{algorithmic}[1]
	\Require{
	\Statex
	\begin{itemize}  	
		\item
  	hws $\ket{\psi^{K}_\text{hws}}$ \Comment{Degree $N_K$~(\ref{Eq:N}) polynomial in bosonic creation operators.}
  	\item
  	$m \in \mathds{Z}^+$ \Comment{$\ket{\psi^{K}_\text{hws}}$ is a hws of SU$(m)$ irrep $K$. }
	\end{itemize} 
	}
	\Ensure{
	\Statex
	\begin{itemize}   	
		\item
  	$\{\Lambda_1,\Lambda_2,\dots, \Lambda_d \colon \Lambda_i \in (\mathds{Z}^+\cup 0)^{m-1}\}$ \Comment {List of weights in the irrep graph of $K$.}
		\item $d,$ Basis sets~(\ref{Eq:BasisSets})
\begin{align}
	&\left\{\ket{\phi^{(1)}_{\Lambda_1}},\ket{\phi^{(2)}_{\Lambda_1}},\dots, \ket{\phi^{(M(\Lambda_1)}_{\Lambda_1}}\right\},
	\left\{\ket{\phi^{(1)}_{\Lambda_2}},\ket{\phi^{(2)}_{\Lambda_2}},\dots, \ket{\phi^{(M(\Lambda_2))}_{\Lambda_2}}\right\},
	\dots,\nonumber \\
	&\left\{\ket{\phi^{(1)}_{\Lambda_i}},\ket{\phi^{(2)}_{\Lambda_i}},\dots, \ket{\phi^{(M(\Lambda_i))}_{\Lambda_i}}\right\},
	\dots,
	\left\{\ket{\phi^{(1)}_{\Lambda_d}},\ket{\phi^{(2)}_{\Lambda_d}},\dots, \ket{\phi^{(M(\Lambda_d))}_{\Lambda_d}}\right\}\nonumber.
\end{align}
	\end{itemize} 
	}

	\Procedure{BasisSet}{$m$, $\ket{\psi^{K}_\text{hws}}$}
	\State 
    Initialize empty statesList, empty weightList and currentStateQueue $\gets \ket{\psi^{K}_\text{hws}}$ \;
	\While{currentStateQueue is not empty}
		\State currentState $\gets$ \textsc{Dequeue}(currentStateQueue)
		\For{\textsc{CurrentOperator} $\in$ set of $\mathfrak{su}(m)$ lowering operations}
			\State newState $\gets$ \textsc{CurrentOperator}(currentState) \label{AlgLin:Lowering}
			\If{newState $\ne$ 0}
				\If {weight of currentState is already in stateList}
					\If{currentState is LI of stateList states with same weight}\label{AlgLin:LI}
						\State independentState $\leftarrow$ \textsc{Normalize}(newState)
						\State {Enqueue} independentState in currentStateQueue
						\State Add independentState to stateList
						\State Add weight of independentState to weightList
					\EndIf\Comment{Else, do nothing.}
				\Else
					\State Enqueue newState in currentStateQueue \label{AlgLine:Enqueue}
					\State Add \{weight(newState),newState\} to stateList
				\EndIf
			\EndIf
		\EndFor
	\EndWhile
	\State Return stateList
	\EndProcedure
	\end{algorithmic}
\end{algorithm}  
\end{figure}

\subsection{Basis-set algorithm (Algorithm~\ref{Algorithm:BasisSet})}
\label{Subsec:BasisSet}
The basis-set algorithm, which finds the basis sets for a given SU$(m)$ irrep, is the key subroutine of our canonical-basis-state algorithm.
Algorithm~\ref{Algorithm:BasisSet} requires inputs $\ket{\psi^{K}_\text{hws}}$ and $m$, where $\ket{\psi^{K}_\text{hws}}$ is a hws of the irrep $K$ of $\mathfrak{su}(m)$ algebra.
The state $\ket{\psi^{K}_\text{hws}}$ is a bosonic state, which is expressed as a summation over products of $N_K$~(\ref{Eq:N}) creation operators $\{a^\dagger_{i,j}\colon i\in \{1,2,\dots,m\}, j\in \{1,2,\dots,m-1\}\}$.
This summation acts on the $m$-site bosonic vacuum state to give an $N_K$-boson state. 
The algorithm returns multiple sets
\begin{align}
	&\left\{\ket{\phi^{(1)}_{\Lambda_1}},\ket{\phi^{(2)}_{\Lambda_1}},\dots, \ket{\phi^{(M(\Lambda_1)}_{\Lambda_1}}\right\},
	\left\{\ket{\phi^{(1)}_{\Lambda_2}},\ket{\phi^{(2)}_{\Lambda_2}},\dots, \ket{\phi^{(M(\Lambda_2))}_{\Lambda_2}}\right\},
	\dots,\nonumber \\
	&\left\{\ket{\phi^{(1)}_{\Lambda_i}},\ket{\phi^{(2)}_{\Lambda_i}},\dots, \ket{\phi^{(M(\Lambda_i))}_{\Lambda_i}}\right\},
	\dots,
	\left\{\ket{\phi^{(1)}_{\Lambda_d}},\ket{\phi^{(2)}_{\Lambda_d}},\dots, \ket{\phi^{(M(\Lambda_d))}_{\Lambda_d}}\right\}
\label{Eq:BasisSets}
\end{align}
of $\mathfrak{su}(m)$ states, with each set spanning the space $\Psi(\Lambda_i)$~(\ref{Eq:Space}) at a different vertex $v(\Lambda_i)$ in the SU$(m)$ irrep $K$.
The states in the output basis sets are represented as polynomials in lowering operators acting on the hws, or equivalently as polynomials in creation and annihilation operators acting on the $n$-site vacuum state. 
Figure~\ref{Figure:Algorithm} is an illustrative example of the algorithm.

A modified breadth-first search (BFS) graph algorithm~\cite{Moore1961,Lee1961,Knuth1998} is used to traverse the irrep graph for states. 
As in usual BFS, we maintain a queue%
\footnote{A queue~\cite{Knuth1998} is a first-in-first-out data structure whose entries are maintained in order.
The two operations allowed on a queue are enqueue, i.e., the addition of entries to the rear and dequeue, which is the removal of entries from the front of the queue.
Both the enqueue and dequeue operations require constant, i.e., O(1) time.}, 
called currentQueue, of the states that have been constructed but whose neighbourhood is yet to be explored.
The algorithm starts with the given hws in currentQueue and iteratively dequeues a state from the front of the queue.
States neighbouring the dequeued state are obtained by enacting one-by-one each of the lowering operators of the algebra.
The newly found states are enqueued into the rear of currentQueue, and the current state and its weight are stored.

We modify BFS to handle vertices with weight multiplicity greater than unity as follows.
While traversing the irrep graph, the algorithm directly enqueues the first state that is found at each vertex.
When the same vertex is explored along a different edge, i.e., by enacting different lowering operators, a different state is found in general.
If the newly constructed state is LI of the states already constructed at the vertex, then the new state is enqueued into currentQueue.%

The algorithm truncates when a state in currentQueue is annihilated by all of the lowering operators and there is no other state in the queue.
This final state must exist because the number of LI states in a given SU$(n)$ irrep is finite according to the following standard result in representation theory.
\begin{lem}[Dimension of an SU$(n)$ irrep~\cite{Slansky1981}]
\label{Lemma:Dimension}
The dimension $\Delta_K$ of the carrier space of an SU$(n)$ irrep $K$ is
\begin{align}
\Delta_{K} \defeq & M(\Lambda_1) + M(\Lambda_2) + \dots + M(\Lambda_d)\nonumber \\
=& \left(1+\kappa_1\right) \left(1+\kappa_2\right)\cdots \left(1+\kappa_{n-1}\right) \left(1+\frac{\kappa_1+\kappa_2}{2}\right)\left(1+\frac{\kappa_2+\kappa_3}{2}\right)\nonumber\\
&\cdots\left(1+\frac{\kappa_{n-2}+\kappa_{n-1}}{2}\right)\left(1+\frac{\kappa_1+\kappa_2+\kappa_3}{3}\right)\left(1+\frac{\kappa_2+\kappa_3+\kappa_4}{3}\right)\nonumber\\
&\cdots\left(1+\frac{\kappa_{n-3}+\kappa_{n-2}+\kappa_{n-1}}{3}\right)\cdots\left(1+\frac{\kappa_1+\kappa_2+\dots+\kappa_{n-1}}{n-1}\right).
\label{Eq:IrrepDimension}
\end{align}
\end{lem}

Now we prove that the basis-set algorithm terminates. 
The proof relies on the fact that the carrier space of SU$(m)$ irrep is finite-dimensional (Lemma~\ref{Lemma:Dimension}).
The algorithm's computational cost is quantified by the number of times the lowering operators are applied on the hws or on states reached by lowering from the hws.
We show that the computational cost of Algorithm~\ref{Algorithm:BasisSet} is linear in the dimension $\Delta_{K}$ of the irrep whose hws is given as input and polynomial in $n$. 
\begin{thm}[Algorithm~\ref{Algorithm:BasisSet} terminates] 
\label{Theorem:1Terminats}
Suppose Algorithm~\ref{Algorithm:BasisSet} receives as input an hws $\ket{\psi^{K}_\text{hws}}$ of an SU$(m)$ irrep $K$. 
Then the algorithm terminates after no more than $\Delta_Km(m-1)/2$ applications of lowering operators.
\end{thm}
\begin{proof}
The proof is in two parts.
Firstly, the number of states that enters currentStateQueue is bounded above by the dimension $\Delta_K$~(\ref{Eq:IrrepDimension}) of the irrep space.
Secondly, as each state that enters currentStateQueue is acted upon by no more than $n(n-1)/2$ lowering operators, the number of lowering operations performed is less than or equal to $\Delta_K n(n-1)/2$.

We show that the number of states that enter currentStateQueue is no more than~$\Delta_K$ as follows.
As each currentState is a linear combination of states obtained by acting lowering operators (Line~\ref{AlgLin:Lowering}) on the given hws, each state that enters currentStateQueue is in the irrep labelled by the hws.
Moreover, each state entering the queue is tested for linear independence (Line~\ref{AlgLin:LI}) with respect to the states already obtained.
Any state that is not LI is discarded.
Therefore, each enqueued state (Algorithm~\ref{Algorithm:BasisSet}, Line~\ref{AlgLine:Enqueue}) is in the correct irrep $K$ and is LI of each other enqueued state.
Thus, the number of states that ever enter currentStateQueue is no more than the number~$\Delta_K$ of LI states in irrep $K$.

In each iteration of the algorithm, we act all the lowering operators on the states in currentStateQueue.
The number of lowering operations is thus bounded above by the product $\Delta_K n(n-1)/2$ of the number of states that enter currentStateQueue and of the number of lowering operators in the $\mathfrak{su}(n)$ algebra.
The algorithm thus terminates after no more than $\Delta_Kn(n-1)/2$ applications of lowering operators.
\end{proof}

We now prove that the algorithm returns the correct output on termination.
The proof requires the following lemma stating that each canonical basis state can be obtained by enacting only with the lowering operators on the hws.
\begin{lem}[Every basis-state can be reached by lowering from the hws~\cite{Humphreys1972}]
\label{Lemma:Cyclic}
No canonical basis-state of a given SU$(n)$ irrep $K$ is LI of the states obtained by lowering from the hws by the action
\begin{equation}
c_{i_{k},j_{k}}\cdots c_{i_{2},j_{2}} c_{i_{1},j_{1}}\ket{\psi_\mathrm{hws}}\quad i_\ell \le j_\ell\, \forall 1\le \ell \le k
\label{Eq:Construction}
\end{equation}
of $k\le \sum_i{\kappa_i}$ number of $\mathfrak{su}(n)$ lowering operators on the hws of the irrep.
\end{lem}
\noindent Lemma~\ref{Lemma:Cyclic} implies that each basis state can be constructed by linearly combining states obtained on lowering from the hws.
Algorithm~\ref{Algorithm:BasisSet} leverages from the construction of Eq.~(\ref{Eq:Construction}) and from testing linear independence to construct the basis sets.

The correctness of the basis-set algorithm is proved as follows.
We show that each state obtained by enacting any number of lowering operators on the hws is LD on the states returned by the algorithm. 
Each canonical basis state is LD on the states obtained by lowering from the hws in turn, so each canonical basis state is LD on the algorithm output.
The algorithm only constructs states in the correct irrep so Algorithm~\ref{Algorithm:BasisSet} returns a complete basis set at each weight of the irrep on truncation.
\begin{thm}[Algorithm~\ref{Algorithm:BasisSet} is correct]
\label{Theorem:1Correct}
The sets
\begin{align}
	&\left\{\ket{\phi^{(1)}_{\Lambda_1}},\ket{\phi^{(2)}_{\Lambda_1}},\dots, \ket{\phi^{(M(\Lambda_1)}_{\Lambda_1}}\right\},
	\left\{\ket{\phi^{(1)}_{\Lambda_2}},\ket{\phi^{(2)}_{\Lambda_2}},\dots, \ket{\phi^{(M(\Lambda_2))}_{\Lambda_2}}\right\},
	\dots,\nonumber \\
	&\left\{\ket{\phi^{(1)}_{\Lambda_i}},\ket{\phi^{(2)}_{\Lambda_i}},\dots, \ket{\phi^{(M(\Lambda_i))}_{\Lambda_i}}\right\},
	\dots,
	\left\{\ket{\phi^{(1)}_{\Lambda_d}},\ket{\phi^{(2)}_{\Lambda_d}},\dots, \ket{\phi^{(M(\Lambda_d))}_{\Lambda_d}}\right\}\nonumber
\end{align} 
of states returned by Algorithm~\ref{Algorithm:BasisSet} span the respective vertex spaces~$\Psi(\Lambda_i)$~(\ref{Eq:Space}) at each vertex~$\Lambda_i$ of the given irrep $K$. 
\end{thm}
\begin{proof}
We first prove by induction that each state in the form of Eq.~(\ref{Eq:Construction}) is LD on states in the algorithm output.
Our induction hypothesis is that each state 
\begin{equation}
c_{i_{k},j_{k}}\cdots c_{i_{2},j_{2}} c_{i_{1},j_{1}}\ket{\psi_\mathrm{hws}},
\label{Eq:EllState}
\end{equation}
which is obtained by acting $\ell$ lowering operators on the hws, is LD on the states returned by the algorithm $\forall\ell\in\mathds{Z^+}$. 
The proof of the hypothesis follows from mathematical induction over $\ell$.

The induction hypothesis is true for base case $k = 1$.
In the first iteration, the algorithm enacts all the lowering operators on the hws~(Algorithm~\ref{Algorithm:BasisSet} Line \ref{AlgLin:Lowering}) and saves each of the obtained states.
No $k = 1$ state~(\ref{Eq:EllState}) is omitted because the vertices neighbouring the hws vertex are all being explored for the first time.
Hence, all the states that can be reached by lowering once from the hws are added to currentStateQueue and, eventually, to stateList.

Assume that the induction hypothesis holds for $k = \ell$, i.e., each $k=\ell$ state is LD on the states in stateList.
We prove that the hypothesis holds for $k = \ell + 1$ by contradiction. 
Suppose there exists a state that can be reached by enacting $\ell+1$ lowering operators on the hws but is LI of stateList.
Let $\ket{\psi} = c_{i_{\ell+1},j_{\ell+1}}c_{i_{\ell},j_{\ell}}\cdots c_{i_{2},j_{2}} c_{i_{1},j_{1}}\ket{\psi_\mathrm{hws}}$ be such a state.

Consider now the state $\ket{\varphi} = c_{i_{\ell},j_{\ell}}\cdots c_{i_{2},j_{2}} c_{i_{1},j_{1}}\ket{\psi_\mathrm{hws}}$ obtained by enacting one less lowering operation from the hws; i.e., $\ket{\psi} = c_{i_{\ell+1},j_{\ell+1}}\ket{\varphi}$.
We have assumed that the induction hypothesis holds for $k = \ell$.
Therefore, $\ket{\varphi}$ is LD on the states constructed by a algorithm.
In other words,
\begin{equation}
\big|{\varphi}\big\rangle = \sum_{j=1}^{J} a_j\ket{\phi_j}
\end{equation}
is LD on the stateList elements $\left\{\ket{\phi_j}:j\in\{1,2,\dots,J\}\right\}$ for complex numbers $a_j$.

The algorithm enacts the lowering operator $c_{i_{\ell+1},j_{\ell+1}}$ on each $\ket{\phi_j}$ and the resulting states are either stored in stateList or are LD on elements in stateList.
Therefore, the elements of the  set $\{c_{i_{\ell+1},j_{\ell+1}}\ket{\phi_j}: j \in\{1,2,\dots,J\}\}$ are LD on the elements of stateList.
Hence, the element $c_{i_{\ell+1},j_{\ell+1}}\ket{\varphi}$ is also LD on the elements of stateList.
This dependence contradicts the supposition that $\ket{\psi} = c_{i_{\ell+1},j_{\ell+1}}\ket{\varphi}$ is LI of stateList, thereby proving the induction hypothesis for $k = \ell+1$.

The induction hypothesis is true for $\ell = 1$ and is shown to hold for $k = \ell + 1$ if it holds for $k = \ell$.
Thus, our induction hypothesis is true for all $\ell\in \mathds{Z}^+$.
Every state obtained of irrep $K$ obtained by lowering from the hws is linearly dependent (LD) on the basis sets that are returned by the algorithm.

We know from Lemma~\ref{Lemma:Cyclic} that each canonical basis state is LD on the states obtained by lowering. 
Hence, each canonical basis state is LD on the states obtained at the output of the algorithm.
Therefore, the state returned by the algorithm span the space of irrep $K$ states, and the output basis sets span the set of all basis states of the given irrep $K$. 
\end{proof}

We have proved that Algorithm~\ref{Algorithm:BasisSet} terminates and that it returns the correct basis sets on termination.
Now we present our algorithm for the construction of the canonical basis states.
Furthermore, we prove the correctness and termination of the canonical-basis-states algorithm.

\subsection{Canonical-basis-states algorithm (Algorithm~\ref{Algorithm:Main})}
\label{Subsec:BasisStates}
The algorithm for constructing the canonical basis-states of SU$(n)$ requires inputs $n \in \mathds{Z}^+$ and the irrep label $K$.
The algorithm returns expressions for all the canonical basis states in the given irrep.
Figure~\ref{Figure:Main} illustrates SU$(3)$ basis-state construction using our algorithm.
Algorithm~\ref{Algorithm:Main} details the step-by-step construction of the canonical basis states.

\begin{figure}[h]
\includegraphics[width=5.5in]{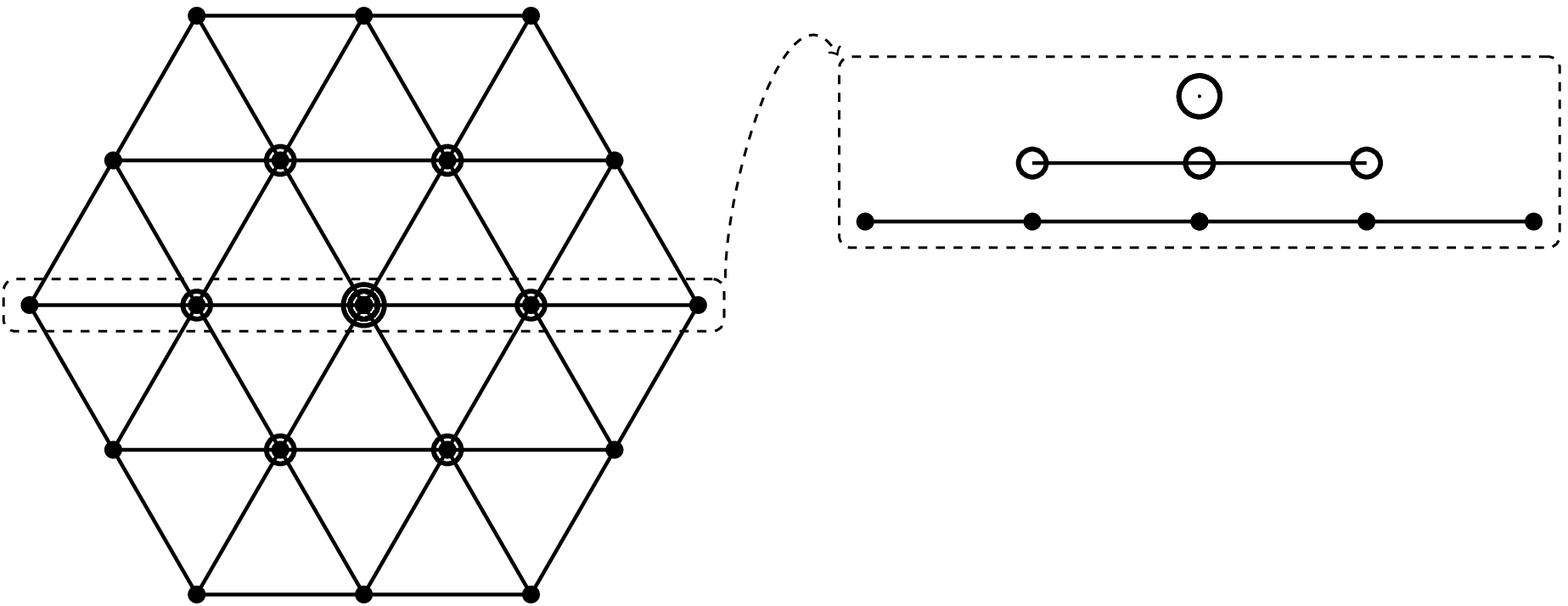}
\caption{Diagrammatic representation of the main algorithm for $n=3$.
The dots and circles represent the canonical basis states.
The dimension of the space of states at a given vertex is the sum of the number of dots and the number of circles at the vertex, for instance weights associated with dimension two are represented by one dot and one circle.
The lines connecting the dots represent the transformation from states of one weight to those of another by the action of SU$(3)$ raising and lowering operators. 
We use Algorithm~\ref{Algorithm:BasisSet} to construct basis sets for each vertex in the SU$(n)$ irrep graph. 
Once the basis sets for the SU$(n)$ irreps are computed, the algorithm enacts the $\mathfrak{su}(n-1)$ raising operators on the $(n-1)$-dimensional sub-irreps to find the $\mathfrak{su}(n-1)$ hws.
Then the algorithm starts with the $\mathfrak{su}(n-1)$ hws and employs the basis-set construction (Algorithm 1) to find all the states in the $\mathfrak{su}(n-1)$ irrep labelled by the hws.
The states thus obtained are subtracted from the set of $\mathfrak{su}(n)$ states.
A new state is chosen from the weight of highest multiplicity and the process repeated until all the $\mathfrak{su}(n-1)$ irreps are found.
}%
\label{Figure:Main}%
\end{figure}

\begin{figure}[h]
\linespread{1}
\begin{algorithm}[H]
\caption{Canonical-basis-states algorithm}\label{Algorithm:Main}
	\begin{algorithmic}[1]
	\Require{
	\Statex
	\begin{itemize}   	
  	\item
  	$n \in \mathds{Z}^+$ \Comment{Algorithm constructs basis sets of $\mathfrak{su}(n)$ algebra.}
		\item 
		$K = (\kappa_1,\kappa_2,\dots,\kappa_{n-1})\in \left(\mathds{Z}^+ \cup \{0\}\right)^{n-1}$ \Comment{Label of SU$(n)$ irrep.}
	\end{itemize} 
	}
	\Ensure{
	\Statex
	\begin{itemize}   	
		\item $\left\{\left(\Big|{\tensor*{\psi}{*^{K^{(n)}}_{\Lambda^{(n)}}^{,\dots,}_{,\dots,}^{K^{(3)},}_{\Lambda^{(3)},}^{K^{(2)}}_{\Lambda^{(2)}}}}\Big\rangle;K^{(n)},\dots,K^{(2)};\Lambda^{(n)},\dots,\Lambda^{(2)}\right)\right\}$\Comment{List of all canonical basis states and weight labels in the irrep $K^{(n)}=K$.}
	\end{itemize} 
	}
	\Statex
	\Procedure{CanonicalBasisStates}{$n$, $K$}
	\State Initialize empty basisStatesList, hws $\leftarrow \ket{\psi^K_\text{hws}}$ \label{Alglin:hws}
	\State SUmStates, SUnStates $\leftarrow$ {\sc BasisSet}($n$,hws) \label{Alglin:FirstStage}
	\While{SUnStates is not empty}
	\For{$m\in \{n,n-1,\dots,2\}$}
		\State $\Lambda_\text{max} \leftarrow \mathfrak{su}(m)$ weight with highest number of states in SUmStates.
		\State $\ket{\psi^{(m)}_\text{max}}\leftarrow$ arbitrary superposition of states at $\Lambda_\text{max}$ in SUmStates.
		\State \label{AlgLin:Raising}Apply $\mathfrak{su}(m-1)$ raising operators on $\ket{\psi^{(m)}_\text{max}}$; reach $\mathfrak{su}(m-1)$ hws $\ket{\psi^{(m-1)}_\text{hws}}$.
		\State $K^{(m-1)} \leftarrow$ {\sc Weight}$\left(\ket{\psi^{(m-1)}_\text{hws}}\right)$.
		\State SUmStates $\leftarrow$ {\sc BasisSet}$\left(m-1,\ket{\psi^{(m-1)}_\text{hws}}\right)$.
		\If{$m=2$}
		\For{All states $\ket{\psi}$ in SUmStates}
			\State $\{\Lambda^{(n)},\dots,\Lambda^{(2)}\} \leftarrow$ {\sc Weights}$\left(\ket{\psi}\right)$ \Comment{$\mathfrak{su}(m)$ weights $\forall m\le n$.} 
			\State Concatenate $\left(\ket{\psi};K^{(n)},\dots,K^{(2)};\Lambda^{(n)},\dots,\Lambda^{(2)} \right)$ basisStatesList 
			\State Subtract SUmStates from SUnStates.
		\EndFor
		\EndIf\Comment{Else, do nothing.}	 
	\EndFor
	\EndWhile
	\For{All states $\ket{\psi^{(i)}}$ in statelist}
		\State Act $\{C_{1,2},C_{2,3},\dots, C_{n-1,n}\}$ on $\ket{\psi^{(i)}}$ until hws $\ket{\psi_\mathrm{hws}^{(i)}}$ is reached.
		\State $\ket{\psi^{(i)}}\leftarrow \e^{i\phi^{(i)}}\ket{\psi^{(i)}}$ for $\ket{\psi_\mathrm{hws}^{(i)}} = \e^{i\phi^{(i)}} \ket{\psi^K_\text{hws}}$.\label{Alglin:PhaseConvention}
	\EndFor
	\State Return basisStatesList
	\EndProcedure
	\end{algorithmic}
\end{algorithm}  
\end{figure}

The canonical-basis-states algorithm proceeds by partitioning $\mathfrak{su}(n)$ basis sets into $\mathfrak{su}(m)$ basis sets for progressively smaller $m$ over $n-1$ stages. 
In the first stage, the algorithm employs Lemma~\ref{Lemma:hws} to construct the hws of the given irrep $K$ (Algorithm~\ref{Algorithm:Main}, Line \ref{Alglin:hws}).
Algorithm~\ref{Algorithm:BasisSet} is then used to construct the basis sets of the SU$(n)$ irrep of the constructed hws (Line \ref{Alglin:FirstStage}).

By the $(n-m)$-th stage, the algorithm has partitioned the entire $\mathfrak{su}(n)$ space into basis states of the SU$(m+1)$ irreps.
In this stage, each of the $\mathfrak{su}(m+1)$ basis sets is partitioned into $\mathfrak{su}(m)$ basis sets by using $\mathfrak{su}(m)$ operators.
The algorithm searches each SU$(m+1)$ irrep graph for the vertex that has the highest multiplicity.

An arbitrary linear combination of the basis states at this vertex is chosen.
The algorithm then enacts all the raising operators in the $\mathfrak{su}(m)$ subalgebra on this linear combination until the action of each of the raising operators annihilates the state.
The state thus obtained is the hws of an SU$(m)$ irrep, whose label $K^{(m)}$ can be calculated by enacting the Cartan operators on the state.

Next the algorithm performs the basis-set construction algorithm on the $\mathfrak{su}(m)$ hws employing only the $\mathfrak{su}(m)$ lowering operators.
This procedure gives us sets of basis states that belong to the SU$(m)$ irrep $K^{(m)}$.
The irrep $K^{(m)}$ basis sets are stored and are then subtracted from the SU$(m+1)$ states.
The algorithm iteratively (i) starts from the highest multiplicity vertex of SU$(m+1)$ irrep graphs, (ii) constructs a hws by raising, (iii) stores the basis sets of SU$(m)$ irreps corresponding to this hws and (iv) subtracts them from SU$(m+1)$ states until all the states in the $\mathfrak{su}(m+1)$ are partitioned.

At the end of $n-1$ stages, we have a list of basis sets of the SU$(n-1)$ irreps.
We iteratively perform the process of finding basis sets for smaller subgroups until we reach SU$(2)$ basis sets, which are known to have unit multiplicity.
Hence, the algorithm returns the basis states that are eigenvectors of the Cartan operators of all SU$(m)\colon m\le n$ groups.

The relative phases between the basis states are fixed by imposing Eq.~(\ref{Eq:PhaseConvention}).
Each of the constructed basis states is acted upon by the simple raising operators $\{C_{1,2},C_{2,3},\dots,C_{n-1,n}\}$ until the hws is reached. 
The phase of this hws obtained by raising is required to be the same for all basis states.
Our algorithm multiplies each of the basis states by a phase factor (Line~\ref{Alglin:PhaseConvention}) to impose the phase convention~Eq.~(\ref{Eq:PhaseConvention}) and returns the set of canonical basis states.

Now we prove that the canonical basis states algorithm terminates.
The proof of termination uses the facts that the number of basis states is equal to the dimension $\Delta_K$ of the irrep and that each basis state is added to currentStateQueue no more than once. 
\begin{thm}[Algorithm~\ref{Algorithm:Main} terminates]
\label{Theorem:2Terminates}
Algorithm~\ref{Algorithm:Main} terminates after the action of no more than $\Delta_K n(n-1)^2/2$ lowering operators
\end{thm}
\begin{proof}
In each of the $n-1$ stages of Algorithm~\ref{Algorithm:Main}, the states that are added to currentStateQueue are LI of each other because of the conditions imposed in the algorithm. 
There are no more LI states in the given SU$(n)$ irrep than the dimension $\Delta_K$ of the irrep space.
Thus, the total number of states that are added to currentQueue in each of the $n-1$ stages is no more than $\Delta_K$.
No more that $n(n-1)/2$ lowering operators are applied on the states that enter currentQueue.
Thus, each stage terminates after the application of $\Delta_K n(n-1)/2$ lowering operations.
Furthermore, the algorithm terminates after $n-1$ stages and the application of no more than $\Delta_K n(n-1)^2/2$ lowering operations.
\end{proof}

Finally, we prove that the canonical-basis-states algorithm returns the correct output when it terminates.
\begin{thm}[Algorithm~\ref{Algorithm:Main} is correct]
\label{Theorem:2Correct}
The SU$(n)$ states
\begin{equation}
\Big|{\tensor*{\psi}{*^{K^{(z)}}_{\Lambda^{(n)}}^{,\dots,}_{,\dots,}^{K^{(3)},}_{\Lambda^{(3)},}^{K^{(2)}}_{\Lambda^{(2)}}}}\Big\rangle
\end{equation}
yielded by Algorithm~\ref{Algorithm:Main} are the canonical states of Definition~\ref{Definition:CanonicalBasisStates}.\end{thm}
\begin{proof}
The theorem holds if the states yielded by Algorithm~\ref{Algorithm:Main} have well defined weights and have well defined irrep labels. 
First we show that the weight of each state in the output of the algorithm is well defined.
Each state in the output is obtained either by enacting lowering operators on the hws or by taking linear combinations of states that have the same weight.
Linear combination of states with the same weights have well defined weights themselves. 
Thus, all the output states have well defined weights for SU$(m)$ irreps for all $2\le m\le n$.

We prove that the states have well defined SU$(m)$ irrep label separately for $m = n$ and for $2\le m\le n-1$.
The correctness of the $\mathfrak{su}(m)$ hws follows from Lemma~\ref{Lemma:hws}.
Every state in the output is a linear combination of states obtained by lowering from the constructed $\mathfrak{su}(m)$ hws.
Thus, every state is in the correct SU$(n)$ irrep $K^{(n)}$.

The algorithm (Line~\ref{AlgLin:Raising}) enacts raising operators on linear combinations of $\mathfrak{su}(m+1)$ basis states at one weight until each of the raising operators annihilates the raised state.
The $\mathfrak{su}(m)$ state thus obtained are legitimate $\mathfrak{su}(m)$ hws's or possibly linear combinations of $\mathfrak{su}(m)$ hws's by construction.
The uniqueness of the hws is guaranteed by the existence of the canonical basis~\cite{Humphreys1972}.
Each of the canonical basis states is obtained by lowering from these $\mathfrak{su}(m)$ hws's using $\mathfrak{su}(m)$ lowering operators and thus have well defined irrep labels for all $2\le m \le n-1$. 

We have shown that the states yielded by the algorithm have well defined values of irrep labels $K^{(m)}$ for $\mathfrak{su}(m)$ algebras for all $\{m:2\le m\le n\}$ and of $\mathfrak{su}(m)$ weights $\Lambda^{(m)}$ for all $\{m:2\le m\le n\}$.
Thus, these states are the canonical SU$(n)$ basis states.
This completes the proof of correctness of Algorithm~\ref{Algorithm:Main}.
\end{proof}

We have proved that Algorithm~\ref{Algorithm:Main} terminates and returns the canonical basis states on termination.
The states constructed by the canonical-basis-states algorithm are employed to compute arbitrary SU$(n)$ $\D$-functions using an algorithm presented in the next subsection.

\subsection{$\D$-function algorithm}
\label{Subsec:D}
\begin{figure}[h]
\linespread{1}
\begin{algorithm}[H]
\caption{$\D$-function Algorithm}\label{Algorithm:D}
	\begin{algorithmic}[1]
	\Require{
	\Statex
	\begin{itemize}  	
  	\item
  	$n \in \mathds{Z}^+$ \Comment{Algorithm constructs $\D$-functions of SU$(n)$ elements.}
		\item
		$\Omega = \{\omega_1,\omega_2,\dots,\omega_{n^2-1}\} \in \mathds{R}^{n^2-1}$ \Comment{Parametrization of SU$(n)$ transformation.} 
		\item
		$K^{(n)},\dots,K^{(2)}$ and ${\Lambda^{(n)},\dots,\Lambda^{(2)}}$ \Comment{Row Label.}
		\item
		$K^{\prime(n)},\dots,K^{\prime(2)}$ and ${\Lambda'^{(n)},\dots,\Lambda'^{(2)}}$ \Comment{Column Label.}
	\end{itemize} 
	}
	\Ensure{
	\Statex
	\begin{itemize}   	
		\item $\tensor*{\D}{*^{K^{(n)}}_{\Lambda^{(n)}}^{,\dots,}_{,\dots,}^{K^{(3)},}_{\Lambda^{(3)},}^{K^{(2)}}_{\Lambda^{(2)}}^;_;^{K^{\prime(n)}}_{\Lambda^{\prime(n)}}^{,\dots,}_{,\dots,}^{K^{\prime(3)},}_{\Lambda^{\prime(3)},}^{K^{\prime(2)}}_{\Lambda^{\prime(2)}}}(\Omega)$	
	\end{itemize} 
	}
	\Statex
	\Procedure{D}{$n,\Omega, K^{(n)},\dots,K^{(2)}, K'^{(n)},\dots,K'^{(2)}, {\Lambda^{(n)},\dots,\Lambda^{(2)}},\Lambda^{\prime (n)},\dots,\Lambda^{\prime (2)}$}
	\State Construct $V \in GL(n,\mathds{C})$ from parametrization $\Omega$~\cite{Reck1994}
	\If{$K^{(n)} = K^{\prime(n)}$}
	\State $\Big|{\tensor*{\psi}{*^{K^{(n)}}_{\Lambda^{(n)}}^{,\dots,}_{,\dots,}^{K^{(3)},}_{\Lambda^{(3)},}^{K^{(2)}}_{\Lambda^{(2)}}}}\Big\rangle \leftarrow$ {using \sc CanonicalBasisStates}($n,K(n)$).	
	\State $\Big|{\tensor*{\psi}{*^{K^{\prime(n)}}_{\Lambda^{\prime(n)}}^{,\dots,}_{,\dots,}^{K^{\prime(3)},}_{\Lambda^{\prime(3)},}^{K^{\prime(2)}}_{\Lambda^{\prime(2)}}}}\Big\rangle \leftarrow$ using {\sc CanonicalBasisStates}($n,K'(n)$).	
	\State Construct $\Big\langle{\tensor*{\psi}{*^{K^{(n)}}_{\Lambda^{(n)}}^{,\dots,}_{,\dots,}^{K^{(3)},}_{\Lambda^{(3)},}^{K^{(2)}}_{\Lambda^{(2)}}}}\Big|$ from $\Big|{\tensor*{\psi}{*^{K^{(n)}}_{\Lambda^{(n)}}^{,\dots,}_{,\dots,}^{K^{(3)},}_{\Lambda^{(3)},}^{K^{(2)}}_{\Lambda^{(2)}}}}\Big\rangle$ by complex conjugation.
	\State Construct $V(\Omega)\Big|{\tensor*{\psi}{*^{K^{\prime(n)}}_{\Lambda^{\prime(n)}}^{,\dots,}_{,\dots,}^{K^{\prime(3)},}_{\Lambda^{\prime(3)},}^{K^{\prime(2)}}_{\Lambda^{\prime(2)}}}}\Big\rangle$ using $a^\dagger_{i,k} \to \sum_j V_{i,j}(\Omega) a^\dagger_{j,k}\,\forall\,a_{i,k},a^\dagger_{i,k}$.	
	\State Return $\D =\Big\langle{\tensor*{\psi}{*^{K^{(n)}}_{\Lambda^{(n)}}^{,\dots,}_{,\dots,}^{K^{(3)},}_{\Lambda^{(3)},}^{K^{(2)}}_{\Lambda^{(2)}}}}\Big|V(\Omega) \Big|{\tensor*{\psi}{*^{K^{\prime(n)}}_{\Lambda^{\prime(n)}}^{,\dots,}_{,\dots,}^{K^{\prime(3)},}_{\Lambda^{\prime(3)},}^{K^{\prime(2)}}_{\Lambda^{\prime(2)}}}}\Big\rangle$
	\Else
	\State Return $\D = 0$
	\EndIf
	\EndProcedure

	\end{algorithmic}
\end{algorithm} 
\end{figure}

Our task is to construct the $\D$-function
\begin{equation}
\tensor*{\D}{*^{K^{(n)}}_{\Lambda^{(n)}}^{,\dots,}_{,\dots,}^{K^{(3)},}_{\Lambda^{(3)},}^{K^{(2)}}_{\Lambda^{(2)}}^;_;^{K^{\prime(n)}}_{\Lambda^{\prime(n)}}^{,\dots,}_{,\dots,}^{K^{\prime(3)},}_{\Lambda^{\prime(3)},}^{K^{\prime(2)}}_{\Lambda^{\prime(2)}}}(\Omega) 
\label{Eq:DDef2}
\end{equation}
for given labels $\{K^{(m)}\}, \{\Lambda^{(m)}\}, \{K'^{(m)}\}, \{\Lambda'^{(m)}\}$ of the SU$(n)$ element $V(\Omega)$ given by the parametrization $\Omega$. 
The $\D$-function~(\ref{Eq:DDef2}) is computed as the inner product between the state
\begin{equation}
\Big|{\tensor*{\psi}{*^{K^{(n)}}_{\Lambda^{(n)}}^{,\dots,}_{,\dots,}^{K^{(3)},}_{\Lambda^{(3)},}^{K^{(2)}}_{\Lambda^{(2)}}}}\Big\rangle
\end{equation}
of Eq.~(\ref{Eq:States}) and the transformed state
\begin{equation}
V(\Omega)\Big|{\tensor*{\psi}{*^{K^{\prime(n)}}_{\Lambda^{\prime(n)}}^{,\dots,}_{,\dots,}^{K^{\prime(3)},}_{\Lambda^{\prime(3)},}^{K^{\prime(2)}}_{\Lambda^{\prime(2)}}}\Big\rangle}.
\end{equation}
Algorithm~\ref{Algorithm:D} constructs the fundamental representation, i.e., the $n\times n$ matrix, $V_{ij}$ of the SU$(n)$ element $V(\Omega)$~\cite{Reck1994}.
Then, the expressions for the basis states
\begin{equation}
\Big|{\tensor*{\psi}{*^{K^{(n)}}_{\Lambda^{(n)}}^{,\dots,}_{,\dots,}^{K^{(3)},}_{\Lambda^{(3)},}^{K^{(2)}}_{\Lambda^{(2)}}}}\Big\rangle, \Big|{\tensor*{\psi}{*^{K^{\prime(n)}}_{\Lambda^{\prime(n)}}^{,\dots,}_{,\dots,}^{K^{\prime(3)},}_{\Lambda^{\prime(3)},}^{K^{\prime(2)}}_{\Lambda^{\prime(2)}}}\Big\rangle
}
\label{Eq:States}
\end{equation}
corresponding to the given labels are computed using the canonical-basis-states algorithm.
The basis states thus obtained are expressed as summations over products of creation and annihilation operators.
$V(\Omega)$ acts on the boson realization by transforming each boson independently according to
\begin{equation}
a^\dagger_{i,j} \rightarrow a^{\dagger\prime}_{i,j} = \sum_{k} V_{ik}(\Omega)a^\dagger_{k,j},
\label{Eq:TransformationU}
\end{equation}
where $\{V_{ik}(\Omega)\}$ are the matrix elements of the $n\times n$ representation of $V(\Omega)$.
The algorithm transforms the second basis state of Eq.~(\ref{Eq:States}) under the action of $V(\Omega)$ by replacing each of the the creation and annihilation operators of the state according to Eq.~(\ref{Eq:TransformationU}).

The $\D$-function is evaluated as the inner product using the commutation relations~(\ref{Eq:ccr}) or equivalently by using the Wick's theorem~\cite{Bogoljubov1959}.
The correctness and termination of Algorithm~\ref{Algorithm:D} follows directly from Theorems~\ref{Theorem:2Terminates} and~\ref{Theorem:2Correct}, which completes our algorithms for the computation of boson realizations of SU$(n)$ states and of $\D$-functions.
 
\section{Conclusion}
\label{Section:Conclusion}

In summary, we have devised an algorithm to compute expressions for boson realizations of the canonical basis states of SU$(n)$ irreps.
Boson realizations are ideally suited for analyzing the physics of single photons, providing a tractable interpretation to basis states as multiphoton states, and to transformations on these states as optical transformations.
Furthermore, we have devised an algorithm to compute expressions for SU$(n)$ $\mathcal{D}$-functions in terms of elements of the fundamental representation of the group.
Our algorithm offers significant advantage over competing algorithms to construct $\mathcal{D}$-functions.
Furthermore, our $\D$-function algorithm lays the groundwork for generalizing the analysis of optical interferometry beyond the three-photon level~\cite{Tan2013,Guise2014,Tillmann2014}.

This work is the first known application of graph-theoretic algorithms to SU$(n)$ representation theory.
We overcome the problem of SU$(n)$ weight multiplicity greater than unity by modifying the breadth-first graph-search algorithm.
Our procedure for generating a basis set can be extended to subgroups of SU$(n)$. 
In particular, the boson realization of the hws of O$(2k)$ and O$(2k+1)$ irreps can be constructed along the lines of Lemma~\ref{Lemma:hws}~\cite{Pang1967,Wong1969,Lohe1971}.
Graph-search algorithms can be employed to construct O$(n)$ basis states and $\D$-functions if the problem of labelling O$(n)$ basis states can be overcome.
Our approach opens the possibility of exploiting the diverse graph-algorithms toolkit for solving problems in representation theory of Lie groups.

\section*{Acknowledgements}
We thank Sandeep K.~Goyal, Abdullah Khalid and Joe Repka for comments.
ID and BCS acknowledge AITF, China Thousand Talents Plan, NSERC and USARO for financial support.
The work of HdeG is supported in part by NSERC.

\appendix
\section{Choice of subalgebra chain}
\label{Appendix:SubAlgebraChoice}
Our algorithms construct canonical basis states that reduce the subalgebra chain~(\ref{Eq:SubalgebraChain}).
Other $\mathfrak{su}(n)\supset\mathfrak{su}(n-1)\supset\dots\supset\mathfrak{su}(2)$ subalgebra chains are possible and our algorithm can be generalized to construct canonical basis states that reduce other chains, as we discuss in this appendix.

Each $\mathfrak{su}(m)$ subalgebra of $\mathfrak{su}(n),\,m<n$ is specified by the sets of raising, lowering and Cartan operators that generate it.
For a given sequence
\begin{equation}
I^{(m)} = \left\{i^{(m)}_1,i^{(m)}_2,\dots, i^{(m)}_{m}\right\}
\end{equation}
of $m$ increasing integers, we can define the corresponding set of raising, lowering and Cartan operators
\begin{align}
&\left\{C_{i_1,i_2},C_{i_1,i_3},\dots ,C_{i_1,i_m},C_{i_2,i_3},\dots ,C_{i_2,i_m},\dots,C_{i_{m-1},i_m}\right\}&\text{(Raising)}\\
&\left\{C_{i_2,i_1},C_{i_3,i_1},\dots ,C_{i_m,i_1},C_{i_3,i_2},\dots ,C_{i_m,i_2},\dots,C_{i_{m},i_{m-1}}\right\}&\text{(Lowering)}\\
&\left\{C_{i_2,i_2}-C_{i_1,i_1},C_{i_3,i_3}-C_{i_2,i_2},\dots,C_{i_m,i_m}-C_{i_{m-1},i_{m-1}}\right\}&\text{(Cartan)}
\end{align}
that generate the algebra.
Thus, each $\mathfrak{su}(n)\supset\mathfrak{su}(n-1)\supset\dots\supset\mathfrak{su}(2)$ subalgebra chain is uniquely specified by the ordered sequences $I^{(m)}\colon m<n$ of integers, where
\begin{align}
\begin{split}
I^{(n-1)} &= \{i^{(n-1)}_1,i^{(n-1)}_2,\dots, i^{(n-1)}_{n-1}\} \subset \{1,2,\dots,n\}\\
I^{(n-2)} &= \{i^{(n-2)}_1,i^{(n-2)}_2,\dots, i^{(n-2)}_{n-1}\} \subset I^{(n-1)}\\
&\dots\\
I^{(m-1)} &= \{i^{(m-1)}_1,i^{(m-1)}_2,\dots, i^{(m-1)}_{m}\} \subset I^{(m)}\\
&\dots\\
I^{(2)} &= \{i^{(2)}_1,i^{(2)}_2\} \subset I^{(3)}.
\label{Eq:Chain}
\end{split}
\end{align}

Consider the example of $\mathfrak{su}(2)$ subalgebras of $\mathfrak{su}(3)$. 
The three subsets
\begin{align}
\{C_{1,2},C_{2,1},C_{1,1}-C_{2,2}\}\label{Eq:A1}\\
\{C_{1,3},C_{3,1},C_{1,1}-C_{3,3}\}\\
\{C_{2,3},C_{3,2},C_{2,2}-C_{3,3}\}\label{Eq:A3}
\end{align}
of the generators $\{C_{i,j}\colon i,j \in \{1,2,3\}\}$ of $\mathfrak{su}(3)$ generate three distinct $\mathfrak{su}(2)$ algebras.
Each of the three subsets~(\ref{Eq:A1})-(\ref{Eq:A3}) can be labelled with a two-element subset of the $\{1,2,3\}$ and can be employed to define canonical basis states of SU$(n)$.
For instance, consider $(\lambda,\kappa) = (1,1)$ irrep of SU$(3)$.
The weight $(\lambda_2,\lambda_1) = (0,0)$ is associated with a two-dimensional space.
We can identify two basis states of this space by specifying the following:
\begin{enumerate}
\item
choice of $\mathfrak{su}(2)$ algebra. For instance $I^{(2)} = \{1,2\}$, which corresponds to the algebra generated by $\{C_{1,2},C_{2,1},C_{1,1}-C_{2,2}\}$,
\item
$\mathfrak{su}_{1,2,3}(3)$ irreps label: $K^{(3)} = (1,1)$, $\mathfrak{su}_{(1,2)}(2)$ irreps label: $K^{(2)} = (0)$ and $(1)$ for the two basis states.
\item
$\mathfrak{su}_{1,2,3}(3)$ weights: $(0,0)$, $\mathfrak{su}_{(1,2)}(2)$ weights: $(0)$.
\end{enumerate}
Another basis set of the $\Lambda =(0,0)$ space of $\mathfrak{su}(3)$ irrep $K = (1,1)$ is specified by choosing a different $\mathfrak{su}(2)$ subalgebra as follows.
\begin{enumerate}
\item
choice of $\mathfrak{su}(2)$ algebra. For instance $I^{(2)} = \{1,3\}$, which corresponds to the algebra generated by $\{C_{1,3},C_{3,1},C_{1,1}-C_{3,3}\}$,
\item
$\mathfrak{su}_{1,2,3}(3)$ irreps label: $K^{(3)} = (1,1)$, $\mathfrak{su}_{(1,3)}(2)$ irreps label: $K^{(2)} = (0)$ and $(1)$ for the two basis states.
\item
$\mathfrak{su}_{1,2,3}(3)$ weights: $(0,0)$, $\mathfrak{su}_{(1,3)}(2)$ weights: $(0)$.
\end{enumerate}
Thus, different choices of subalgebra chain give us different basis states.

In the main text, we have chosen the subalgebra chain~(\ref{Eq:SubalgebraChain}).
Our algorithms can be modified to account for other choices of subalgbra chain by choosing a different set of lowering operators in the basis-set subroutine.
Thus our algorithms can be used to construct states and $\D$-functions in any of the bases that reduce $\mathfrak{su}(m)$ subalgebra chains.

\section{Connection to Gelfand-Tsetlin basis}
\label{Appendix:Connection}
In this appendix, we detail the mapping between our SU$(n)$ basis states and the canonical Gelfand-Tsetlin (GT) basis.
The GT basis identifies each SU$(n)$ irrep with a sequence of $n$ numbers 
\begin{align}
S_n &= (m_{1,n},\dots,m_{n,n})\\
m_{k,n} &\ge m_{k+1,n}~\forall 1\le k\le n-1,
\end{align}
where the first label in the subscript is the sequence index and the second label identifies the algebra.
The carrier space of every $\mathfrak{su}(m)$ subalgebra is composed of disjoint $\mathfrak{su}(m-1)$ carrier spaces
\begin{equation}
\left\{(m_{1,n-1},\dots,m_{n-1,n-1})\right\}
\end{equation}
that obey the betweenness condition
\begin{equation}
m_{k,n}\ge m_{k,n-1}\ge m_{k+1,n}.
\end{equation}
Thus, each $\mathfrak{su}(n)$ basis state $\ket{M}$ can be labelled by the GT pattern
\begin{equation}
\ket{M} \equiv \begin{pmatrix}
\multicolumn{2}{c}{m_{1,N}} & \multicolumn{2}{c}{m_{2,N}} &
\multicolumn{2}{c}{\ldots} & \multicolumn{2}{c}{m_{N,N}} \\
& \multicolumn{2}{c}{m_{1,N-1}} & \multicolumn{2}{c}{\ldots}
& \multicolumn{2}{c}{m_{N-1,N-1}} & \\
&& \ddots &&& \reflectbox{\(\ddots\)} && \\
&& \multicolumn{2}{r}{m_{1,2}} & \multicolumn{2}{l}{m_{2,2}} && \\
&&& \multicolumn{2}{c}{m_{1,1}} &&&
\end{pmatrix},
\end{equation}
where
\begin{equation}
m_{k,\ell}\ge m_{k,n-1}\ge m_{k+1,\ell}\,, \,1\le k < \ell \le n.
\end{equation}

The canonical basis states are eigenstates of the Cartan operators $\{H_i\}$~(\ref{Eq:Cartan}) as detailed in the following lemma.
\begin{lem}[Connection to Gelfand-Tsetlin basis~\cite{Alex2011}] The canonical basis states are connected to the GT basis according to 
\begin{equation}
\label{eq:gtpattern}
\Big|{\tensor*{\psi}{*^{K^{(z)}}_{\Lambda^{(n)}}^{,\dots,}_{,\dots,}^{K^{(3)},}_{\Lambda^{(3)},}^{K^{(2)}}_{\Lambda^{(2)}}}}\Big\rangle = 
 \begin{pmatrix}
\multicolumn{2}{c}{m_{1,N}} & \multicolumn{2}{c}{m_{2,N}} &
\multicolumn{2}{c}{\ldots} & \multicolumn{2}{c}{m_{N,N}} \\
& \multicolumn{2}{c}{m_{1,N-1}} & \multicolumn{2}{c}{\ldots}
& \multicolumn{2}{c}{m_{N-1,N-1}} & \\
&& \ddots &&& \reflectbox{\(\ddots\)} && \\
&& \multicolumn{2}{r}{m_{1,2}} & \multicolumn{2}{l}{m_{2,2}} && \\
&&& \multicolumn{2}{c}{m_{1,1}} &&&
\end{pmatrix}
\end{equation}
Every state $\ket{M}$ in the GT-labeling scheme is a simultaneous eigenstate of all $\mathfrak{su}(n)$ Cartan operators,
\begin{equation}
H_\ell \ket{M} = \lambda^M_\ell \ket{M}, \quad (1 \le \ell \le N-1),
\end{equation}
with eigenvalues
\begin{equation}
\label{eq:jzelement}
\lambda_\ell = \sum_{k=1}^\ell m_{k,\ell} - \frac{1}{2} \left(\sum_{k=1}^{\ell+1} m_{k,\ell+1}+ \sum_{k=1}^{\ell-1} m_{k,\ell-1}\right), \, 1 \leq \ell \leq N-1.
\end{equation}
\end{lem}
\noindent 
Thus, the canonical basis states of Def.~\ref{Definition:CanonicalBasisStates} is uniquely mapped to the GT basis.

Furthermore, the weights $\lambda_{\ell}$ are also mapped via the boson realizations to differences in number of bosons at sites $\ell$ and $\ell+1$.
Hence, the difference
\begin{equation}
\nu_{\ell+1}- \nu_{\ell} = \sum_{k=1}^l m_{k,\ell} - \frac{1}{2} \left(\sum_{k=1}^{\ell+1} m_{k,\ell+1}+ \sum_{k=1}^{\ell-1} m_{k,\ell-1}\right), \,1 \leq \ell \leq N-1
\end{equation}
in the number of bosons at sites $\ell+1$ and $\ell$ of the boson realization of a basis state is also connected to its GT pattern.
Once we recall the total number of bosons in the system is $\nu_{1}+\nu_{2}+\dots+\nu_{n}=N_k$, one can then invert the differences and recover $\nu_p$ in term of the $m_{k,\ell-1}.$
Thus the canonical GT basis states are connected to our SU$(n)$ basis states.

\bibliography{BosonRealization}
\end{document}